\documentclass[12pt]{amsart}

\usepackage{fullpage, graphicx, amsmath, amssymb, tikz, listings, shuffle }
\newtheorem{thm}{Theorem}[section]
\newtheorem{prop}[thm]{Proposition}

\newtheorem{lemma}[thm]{Lemma}

\newtheorem{definition}[thm]{Definition}
\theoremstyle{definition}
\newtheorem{example}[thm]{Example}

\newtheorem{remark}[thm]{Remark}

\newcommand{\ahat}[1] {\underline{\hat{a}}_{#1}}
\newcommand{\rter}[1] {\mathcal{R}_{\{#1\}}}

\title{Generalized chord diagram expansions of Dyson-Schwinger equations}

\author{Markus Hihn and Karen Yeats}
\thanks{MH was supported by the Humboldt Foundation and wants to thank Dirk Kreimer for his help.  KY is supported by an NSERC discovery grant.}

\begin{document}

\begin{abstract}
Series solutions for a large family of single equation Dyson-Schwinger equations are given as expansions over decorated rooted connected chord diagrams.  The analytic input to the new expansions are the expansions of the regularized integrals for the primitive graphs building the Dyson-Schwinger equation.  Each decorated chord diagram contributes a weighted monomial in the coefficients of the expansions of the primitives and so indexes the analytic solution in a tightly controlled way.
\end{abstract}

\maketitle

\section{Introduction}

Dyson-Schwinger equations are integral equations in quantum field theory; they correspond to the classical equations of motion and so are physically highly meaningful and important.  On the more mathematical side, Dyson-Schwinger equations have a recursive structure which mirrors the decomposition of Feynman diagrams into subdiagrams.  This means that Dyson-Schwinger equations act similarly to functional equations satisfied by generating functions, and so combinatorial tools are useful in understanding them.  

A first step in this direction is strictly diagrammatic, and views Dyson-Schwinger equations as equations which recursively generate the Feynman diagrams themselves, or similarly rooted trees representing their subdiagram structure.  This view neglects the analytic side -- the diagrams still need to be evaluated, but is already interesting, and has been pursued by Lo\"ic Foissy \cite{Fdse, Fphysical}.

The next step is to incorporate the analytic information.  We would like to do so while maintaining a combinatorial understanding of the objects.  In \cite{kythesis} (also available as \cite{Ymem} with an updated final chapter) the second one of us discussed a transformation of a class of analytic Dyson-Schwinger equations into a different form which is well suited for this kind of treatment, see \eqref{DSE}.    In \cite{MYchord} the second author, along with Nicolas Marie, gave the series solution to a particular, fairly narrow class of such analytic Dyson-Schwinger equations as an expansion indexed by rooted connected chord diagrams, where each chord diagram contributes a single monomial.  This expansion was novel and unexpected, hence somewhat interesting.  However, it was unclear to what extent it could be generalized to a broader class of Dyson-Schwinger equations, hence it was unclear whether or not this chord diagram expansion was just a peculiarity or how much it might or might not be telling us something actually interesting for physics.

In this paper, which is based on the PhD thesis of the first author \cite{Hphd}, we generalize the chord diagram expansion of \cite{MYchord} to a substantially larger class of Dyson-Schwinger equations which includes the form of typical single equation, single scale Dyson-Schwinger equations in physics.  The main result, Theorem~\ref{main thm} is a series solution to any Dyson-Schwinger equation of this class.  As in the special case studied before, this expansion is indexed by rooted connected chord diagrams with each contributing a single monomial.  The difference is that the chord diagrams are now decorated with the set of possible decorations determined by shape of the Dyson-Schwinger equation, and the monomials come with a weight depending on the chord diagram and its decorations.

\medskip

The structure of this paper is as follows.  First we will briefly discuss the Dyson-Schwinger set up in Section~\ref{sec DSE}.  Section~\ref{sec chord} defines rooted connected chord diagrams and the features of them which we will need.  Section~\ref{sec tree} defines the insertion tree of a chord diagram and the weight of a chord diagram.   Note that these insertion trees are different from the insertion trees of Feynman graphs which are trees which capture the subdivergence structure of Feynman graphs.  Section~\ref{sec shuffle} investigates the insertion trees in more detail answering the question of how the labels of two subtrees can be consistently combined into one tree.  Section~\ref{sec:diamond} looks at the decomposition of chord diagrams coming from decomposing the corresponding trees into the two subtrees given by the children of the root.  The main result is presented and proved in Section~\ref{sec main}, and the paper ends with a brief conclusion.

\section{Dyson-Schwinger equations}\label{sec DSE}

Suppose we begin with a Dyson-Schwinger equation which is more or less in a recognizable physics form, having only nonstandard notation and normalization, for example
\begin{equation}\label{Yukawa eq orig}
G(x,L) = 1 - \frac{x}{q^2}\int d^4k\frac{k\cdot q}{k^2G(x, \log k^2)(k+q)^2} - \cdots \bigg|_{q^2=\mu^2}
\end{equation}
This is the Dyson-Schwinger equation for the part of the massless fermion self-energy in Yukawa theory which is formed by inserting 
\includegraphics{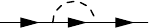} into itself iteratively in all possible ways.  In \eqref{Yukawa eq orig} $x$ is playing the role of the coupling constant, $q$ is the momentum going through, renormalization is taking place by subtraction at a fixed reference scale $\mu$, and $G(x,L)$ is the fermion Green function.  This Dyson-Schwinger equation was solved by Broadhurst and Kreimer in \cite{BKerfc}.

Now suppose we expand $G(x, L)$ in $L$, convert logarithms to powers using $\frac{d^k}{d\rho^k} y^{\rho} |_{\rho=0} = \log^k(y)$,  swap the order of the operators, and recombine the expansion.  Then we obtain
\begin{equation}\label{Yukawa eg mod}
G(x,L) = 1 - xG(x, \partial_{-\rho})^{-1}(e^{-L\rho}-1)F(\rho)|_{\rho=0}
\end{equation}
where $F(\rho)$ is the Feynman integral for the primitive\footnote{This Feynman graph is primitive in the appropriate renormalization Hopf algebra, hence the term \emph{primitive}.  This algebraic framework, while it underpins the entire theory to which this paper contributes, will not be used directly and so will not be defined.  The reader can see \cite{ck0, e-fksurvey, vS, vS2} to read more about it.}, \includegraphics{Yukawaeg} in this case, regularized on the propagator we are inserting at and the integral evaluated at $q^2 = 1$.  In this case $F(\rho) = 1/(\rho)(1-\rho)$.  Example 3.7 of \cite{kythesis} shows this example in detail.

In \cite{MYchord}, \eqref{Yukawa eg mod} is solved as an expansion over rooted connected chord diagrams.  However, the form of \eqref{Yukawa eg mod} is quite specialized.  First of all there is only one primitive Feynman graph.  Second $G(x,L)$ appears once in the denominator of the integrand of the original Dyson-Schwinger equation.  This corresponds to the fact that there is one insertion place.  The more general equation which we will solve here is
\begin{equation}\label{DSE}
G(x,L) = 1 - \sum_{k\geq 1}x^kG(x, \partial_{-\rho})^{1-sk}(e^{-L\rho}-1)F_k(\rho)|_{\rho=0}
\end{equation}
In this equation, $k$ indexes the loop numbers of the primitives.  $s$ is a parameter indicating the degree of the insertion.  The $F_k(\rho)$ are again the regularized Feynman integrals of the primitives.  From now on we will view the $F_k(\rho)$ as given.  Since we are interested in series solutions we will always think of them in terms of their expansions which we assume to have a simple pole at $0$.  We will use the notation
\[
F_k(\rho) = \sum_{i \geq 0} a_{k,i}\rho^{i-1}
\]
for the expansions of the $F_k$.

The previous example, \eqref{Yukawa eg mod}, was the $k=1$, $s=2$ case of \eqref{DSE}.  The photon in quantum electrodynamics would be an $s=1$ case.  We can see this because at 1 loop ($k=1$) there is no insertion place for a photon in the one loop photon correction, for any two loop photon correction there is one insertion place for a photon, and the sequence continues as described above.  See \cite{kythesis} Section 3.3.2 for a discussion of the combinatorics of counting insertions.

Returning to the idea of combinatorial Dyson-Schwinger equations mentioned in the introduction, a good way to think about Dyson-Schwinger equations as in \eqref{DSE}, is to first have a combinatorial Dyson-Schwinger equation in one variable $x$, which captures how the contributing Feynman diagrams (or trees) are formed recursively.  Then the combinatorial Dyson-Schwinger equation can be directly translated into the form of \eqref{DSE} which we call the associated analytic Dyson-Schwinger equation.  See \cite{kythesis} for details.

\medskip 

From a mathematical perspective the problem we will solve in this paper is the following.  Suppose $s$ and the expansions of the $F_k$ are given.  The Dyson-Schwinger equation \eqref{DSE} recursively defines the coefficients of $G(x,L)$ viewed as a bivariate power series in $x$ and $L$.  However, it does not do so in a nice or easy to use way.  We will give an explicit, combinatorial description of the series $G(x,L)$ solving \eqref{DSE}, see Theorem~\ref{main thm}.  This gives the solution to this substantial class of Dyson-Schwinger equations as a kind of weighted generating function of chord diagrams.

There is one property of the series expansion of $G(x,L)$ which we can see directly from the Dyson-Schwinger equation.  This is renormalization group equation translated into this context.

\begin{prop}\label{prop:diff rg eqn}
Let $G(x,L)$ satisfy \eqref{DSE}.  Suppose 
\[
G(x,L) = 1 - \sum_{i\geq 1}L^i\gamma_i(x) \quad \text{and} \quad \gamma_i(x) = \sum_{j \geq i} \gamma_{i,j}x^j
\]
Then
\[
\gamma_k(x) = \frac{1}{k}\gamma_1(x) \left(1 - sx\frac{d}{dx}\right)\gamma_{k-1}(x) 
\]
\end{prop}

\begin{proof}
Since $G$ satisfies \eqref{DSE} it satisfies the renormalization group equation.  Extracting coefficients of $L$ gives the result.  The result can also be proved algebraically by using the Dynkin operator $S\star Y$.  See \cite{kythesis} or \cite{Ymem} chapter 4 for details of both approaches.
\end{proof}

\section{Chord diagrams}\label{sec chord}


\begin{definition}
  \label{def:chord-diagram-alternate}
  A rooted chord diagram $D$ of size $n$ is a fixed point free involution $D \in
  S_{2n}$; that is, a permutation such that $D^2 = \text{id}_{S_{2n}}$ with
  $D(i)\neq i \text{ for all }i = 1 \ldots 2n$.
  Equivalently $D$ is a permutation that can be written as product of disjoint
  transpositions without fixed points:
  \begin{align*}
    D = (x_1 y_1) (x_2 y_2) \cdots (x_n y_n)
  \end{align*}
  where $x_1 < \ldots < x_n$ and $x_i < y_i$ for all $i = 1\ldots 2n$. Each transposition is called a chord and $(x_1 y_1)$ is called the root chord.
\end{definition}

We can visualize a rooted chord diagram as follows. Draw a
circle and mark $2n$ vertices on it. Choose a distinct vertex as the root and label that vertex $1$.
Enumerate the vertices counter-clockwise and draw a chord between vertex $x_i$ and vertex $y_i$
for each transposition. 

We say that a chord $(x_i y_i)$ crosses a chord $(x_j y_j)$ if 
\begin{align*}
  x_i < x_j < y_i < y_j \text{ or } x_{j} < x_i < y_j < y_i
\end{align*}.

\begin{definition}
Let $C$ be a rooted chord diagram.
\begin{itemize}
\item The intersection graph of $C$ is the graph with a vertex for each chord, where the vertex is labeled $i$ for chord $(x_i y_i)$, and with two vertices joined by an edge if the corresponding chords cross.
\item $C$ is connected if its intersection graph is connected
\item The terminal chords of $C$ are those chords $(x_i, y_i)$ which do not cross any chords $(x_j, y_j)$ with $i<j$.  The set of terminal chords of $C$ is denoted $\text{ter}(C)$. 
\end{itemize}
\end{definition}

\begin{example}
A rooted connected chord diagram $C$ with its intersection graph $\Gamma(C)$:
\begin{align*}
C = 
\begin{tikzpicture}[baseline={([yshift=-1.5ex]current bounding box.center)}]\tikzstyle{every node}=[font=\tiny]
\draw[black,thin](0,0) circle(0.75);
\foreach \a in {0,36,...,359}
	\draw[fill=black] (\a:0.75) circle(0.3mm);
\draw (0:0.75) circle(0.9mm);
\draw[thin,black](0:0.75) -- (144:0.75);
\draw(0:1.05) node{$1$};\draw[thin,black](36:0.75) -- (324:0.75);
\draw(36:1.05) node{$2$};\draw[thin,black](72:0.75) -- (288:0.75);
\draw(72:1.05) node{$3$};\draw[thin,black](108:0.75) -- (216:0.75);
\draw(108:1.05) node{$4$};\draw[thin,black](180:0.75) -- (252:0.75);
\draw(180:1.05) node{$5$};\end{tikzpicture}\qquad \Gamma(C) = 
\begin{tikzpicture}[baseline={([yshift=-1.5ex]current bounding box.center)}]\tikzstyle{every node}=[font=\tiny]
\foreach \a in {0,72,...,359}
	\draw[fill=black] (\a:0.75) circle(0.3mm);
\draw (0:0.75) circle(0.9mm);
\draw(0:1.05) node{$1$};\draw(72:1.05) node{$2$};\draw(144:1.05) node{$3$};\draw(216:1.05) node{$4$};\draw(288:1.05) node{$5$};\draw[thin,black](0:0.75) -- (72:0.75);
\draw[thin,black](0:0.75) -- (144:0.75);
\draw[thin,black](0:0.75) -- (216:0.75);
\draw[thin,black](216:0.75) -- (288:0.75);
\end{tikzpicture}
\end{align*}
\end{example}

We will be interested in rooted connected chord diagrams where each chord is assigned a decoration from a set of possible decorations.  Such a chord diagram is called a decorated chord diagram.  

\begin{definition}
\mbox{}
\begin{itemize}
\item Let $\mathcal{R}$ be the set of rooted connected chord diagrams.  
\item Let $\mathcal{R}_n$ be the set of rooted connected chord diagrams with $n$ chords.  
\item Let $\mathcal{R}_{\text{dec}}$ be the set of all decorated rooted connected chord diagrams with the chord decorations from a set $D$. 
\item If $c$ is a chord of a decorated chord diagram we will write $d(c)$ for the decoration of $c$.
\end{itemize}
\end{definition}

The order we need for our constructions is not the obvious counterclockwise order, but rather a different order which we define recursively as follows.

\begin{definition}[Intersection order]
  The intersection order of a rooted chord diagram $C$ is defined recursively by the following pseudo code:
\begin{verbatim}
  intersection_order(k, C) {
    m := root(C)
    label(m) := k
    k := k + 1
    if |C| != 1 then
    foreach D := connected_components(C \ m) traversed counter clockwise
    {
      intersection_order(k,D) 
      k := k + |D|
    }
  }
\end{verbatim}
and start the procedure with \begin{verbatim}intersection_order(1,C).\end{verbatim}
\end{definition}
The following picture shows a chord diagram where its chords are labeled in the intersection order.  The root chord is indicated with a dotted line.
\begin{align*}
\begin{tikzpicture}[baseline={([yshift=-2.4ex]current bounding box.center)}]\tikzstyle{every node}=[font=\small]
\draw[black,thin](0,0) circle(2);
\draw[fill=black](270:2.0) circle(0.5mm); 
\draw(270:2.0) circle(1.0mm);
\draw[fill=black](90:2.0)  circle(0.5mm); 
\draw[fill=black](80:2.0)  circle(0.5mm); 
\draw[fill=black](100:2.0)  circle(0.5mm); 
\draw[fill=black](110:2.0)  circle(0.5mm); 
\draw[fill=black](120:2.0)  circle(0.5mm); 
\draw[fill=black](130:2.0)  circle(0.5mm); 
\draw[fill=black](140:2.0)  circle(0.5mm); 
\draw[fill=black](160:2.0)  circle(0.5mm); 
\draw[fill=black](170:2.0)  circle(0.5mm); 
\draw[fill=black](180:2.0)  circle(0.5mm); 
\draw[fill=black](200:2.0)  circle(0.5mm); 
\draw[fill=black](210:2.0)  circle(0.5mm); 
\draw[fill=black](220:2.0)  circle(0.5mm); 
\draw[fill=black](230:2.0)  circle(0.5mm); 
\draw[fill=black](240:2.0)  circle(0.5mm); 
\draw[fill=black](330:2.0)  circle(0.5mm); 
\draw[fill=black](30:2.0)  circle(0.5mm); 
\draw[thin,dashed](270:2.0) -- (90:2.0); 
\draw[thin,black](330:2.0) -- (170:2.0); 
\draw[thin,black](30:2.0) -- (210:2.0); 
\draw plot [smooth] coordinates {(80:2.0) (90:1.8) (100:1.8) (110:2.0)};
\draw plot [smooth] coordinates {(100:2.0) (110:1.8) (120:1.8) (130:2.0)};
\draw plot [smooth] coordinates {(120:2.0) (130:1.8) (140:2.0)};
\draw plot [smooth] coordinates {(160:2.0) (170:1.8) (180:2.0)};
\draw plot [smooth] coordinates {(200:2.0) (210:1.8) (220:1.8) (230:2.0)};
\draw plot [smooth] coordinates {(220:2.0) (230:1.8) (240:2.0)};
\draw(270:2.4) node{1};
\draw(330:2.3) node{2};
\draw(30:2.3) node{3};
\draw(195:2.2) node{4};
\draw(218:2.2) node{5};
\draw(160:2.2) node{6};
\draw(80:2.3) node{7};
\draw(100:2.3) node{8};
\draw(120:2.3) node{9};
\end{tikzpicture}
\end{align*}

From now on always assume chords are labeled in intersection order.
\begin{definition}
Let $Ter(C) = \{ t_0 < \ldots < t_n \}$
and $d_k$ be the decoration  of the $k$-th chord, then we define:
\begin{align*}
& ||C||:= \sum_{c = 1}^{|C|} d_c\\
& {\underline {\hat a}}_{C} := \left(\prod_{c=1}^{n} a_{d_{t_{c}},t_{c}-t_{c-1}}\right) \cdot \left(\prod_{k\not\in \text{ter(C)}} a_{d_k, 0}\right)
\end{align*}
The symbol $b(C)$ denotes the label of the base chord which is defined to be the smallest terminal chord.
\end{definition}
Note that the hat notation in the above definition does not denote an operator but only that the monomial does not contain $a_{d_{b(C)},b(C)-k}$.

\begin{definition}[Insertion operation for rooted chord diagrams]
  Let $C\in \mathcal{R}(n), D\in \mathcal{R}(m)$ and let 
  \begin{align*}
    C = (x_1 y_1) \ldots (x_n y_n) \\
    D = (x'_1 y'_1) \ldots (x'_m y'_m)
  \end{align*}
their underlying permutations.
Then for each $k=1,\ldots,2m-1$ define $C\circ_k D \in \mathcal{R}(n+m)$ by the following permutation:
  \begin{align*}
   (x_1,y_1+k).. (x_n+k,y_n+k) (H_{n,k}(x'_1),H_{n,k}(y'_1)) .. (H_{n,k}(x'_m),H_{n,k}(y'_m))
  \end{align*}
where $H_{n,k}$ is defined to be:
\begin{align*}
  H_{n,k}(x) = \begin{cases} x+1 & \text{ if } x \leq k \\ x+2n & \text{otherwise}\end{cases}
\end{align*}
\end{definition}

Intuitively what this insertion operation does is put the root of $C$ immediately before the root of $D$ and then put the rest of $C$ into the $k$th interval between the ends of the chords of $D$.

$C \circ_k D$ is indeed a rooted connected chord diagram of size $n+m$, because every integer from one to $2(n+m)$ appears exactly once in the transpositions, it is fix point free and the labeling shift does not destroy any crossings but the root chord of $C$ creates at least one new crossing with chords of $D$.

\begin{example}
  \begin{align*}
  (1,4)(2,6)(3,5) \circ_2 (1,4)(2,5)(3,6) = (1,6)(2,10)(3,11)(4,8)(5,7)(9,12)   
  \end{align*}
\end{example}

\begin{remark}
  Note that the insertion operation is highly non associative and non commutative. If $C \circ_k D$ is defined, $D \circ_k C$ may not be defined. For example, if $C$ is a chord diagram with only one chord, then $D\circ_k C$ is not defined for $k\geq 2$ because there is only one insertion interval in $C$ by definition.
\end{remark}


In the following definition of the root share decomposition, we will need to decompose chord diagrams. Since by our definition chord diagrams are certain permutations, we need to define for an expression 
\begin{align*}
  A = (a_1 a_2) \ldots (a_{2n-1} a_{2n}) 
\end{align*}
where $a_k \; (k=1..2n)$ are arbitrary distinct natural numbers, the associated normalized involution $\text{norm}(A)$ by
\begin{align*}
  \text{norm}(A) = \big(\sigma(1) \sigma(2)\big) \ldots \big(\sigma(2n-1) \sigma(2n)\big)
\end{align*}
where $\sigma \in S_{2n}:$ $a_{\sigma^{-1}(1)} < a_{\sigma^{-1}(2)} \ldots < a_{\sigma^{-1}(2n)}$.

\begin{example}
  Let $A = (1 3)(2 8)(5 7)$, then $\text{norm}(A) = (1 3)(2 6)(4 5)$
\end{example}

\begin{definition}[Root share decomposition]
  Let $C\in \mathcal{R}$ with $|C| > 1$, there exists a unique $i$ such that
  \begin{align*}
    & C = C' \circ_i C'' \text{ where } C' = \text{norm}(C \setminus C_1) \text{, } C'' = \text{norm}(C_1)
  \end{align*}
 and $C_1$ is the first connected component of $C$ with the root chord removed. Note that $C\setminus C_1$ is always connected. This decomposition is called the root share decomposition of $C$.
\end{definition}

\section{Insertion trees and weights}\label{sec tree}

We will now need to associate a binary rooted tree to each rooted connected chord diagram.  The leaves of the tree will correspond to the chords and will be labeled accordingly.  This same construction was used in \cite{MYchord}, however, there the construction appeared to only be a technical tool to prove a certain recurrence.  In the more generalized case discussed here, the tree is used to define the weight with which each chord diagram contributes to the Green function.  Thus, the rooted trees are not merely technical, but actually capture a key part of the structure which we need in the general case.

\begin{definition}[Insertion operation on rooted plane trees]\label{def:instree}
  Let $T,T'$ be rooted plane trees with a virtual edge above the root.  Label the virtual edge 1 and label the remaining edges following a pre-order traversal. The rooted plane tree $T\circ_kT'$ is defined by 
\begin{itemize}
  \item putting a new vertex in the middle of edge $k$ of $T'$, 
  \item placing $T$ as the left subtree rooted at this vertex and 
  \item placing the subtree of $T'$ rooted at the bottom end of $k$ as the right subtree of the new vertex.
\end{itemize}
\end{definition}

\begin{example}
Let $S = \begin{tikzpicture}[baseline={([yshift=0.0ex]current bounding box.center)},level/.style={sibling distance=7mm,level distance=7mm}, level 2/.style={sibling distance=7mm}]
\tikzstyle{every node}=[font=\small]  
  \coordinate
 child{node{}} child{ child{node{}} child{node{}} } ;
\end{tikzpicture}$ and $T = \begin{tikzpicture}[baseline={([yshift=0.0ex]current bounding box.center)},level/.style={sibling distance=7mm,level distance=7mm}, level 2/.style={sibling distance=7mm}]
\tikzstyle{every node}=[font=\small]  
  \coordinate
  child{node{ }} child{node{ }};
\end{tikzpicture}$ then 
 $S\circ_2 T=  \begin{tikzpicture}[baseline={([yshift=0.4ex]current bounding box.center)},level/.style={sibling distance=7mm,level distance=5mm}, level 2/.style={sibling distance=5mm}]
\tikzstyle{every node}=[font=\tiny]
\coordinate
 child{ child{ child{node{ }}child{ child{node{ }}child{node{ }}}}child{node{ }}}child{node{ }};\end{tikzpicture}$
\end{example}

Recall that chord diagrams are labeled by the intersection order.
\begin{definition}[The binary tree $T(C)$ associated to $C$]
  Let $C = C_1 \circ_k C_2$ decomposed by the root share decomposition and let $C_1$, $C_2$ be labeled by the induced labeling of $C$. Then $T(C)$ is defined recursively by 
  \begin{align*}
    T(C) = \begin{cases} \text{one vertex labeled by l} & \text{if } |C|=1 \text{ and is labeled by } l \text{ in the induced labeling} \\
   T(C_1) \circ_k T(C_2) & \text{when } C = C_{1} \circ_k C_2
   \end{cases}
  \end{align*}
\end{definition}

\begin{definition}
Let $v$ be a leaf of a binary rooted tree.  Consider the path beginning at $v$ and moving up and to the left as long as such an edge exists.  Define $\nu_v$ to be the number of edges in this path.
\end{definition}

Note that if a leaf is a left child then $\nu_v=0$.  For a more precise definition of $\nu_v$ using the binary string representation of a binary rooted tree see \cite{Hphd}.

\begin{definition}[Branch-left vector]
  \label{def:branch-left-vector}
  Let $C$ be a rooted, connected chord diagram of size $n$ and $T(C)$ its corresponding unique insertion tree, then $\nu(C) = (\nu_1,\ldots,\nu_n)$ is said to be the branch-left vector of $C$.
\end{definition}



Now, we can define the weight mentioned in the introduction of this section. 
\begin{definition}[The weight of a decorated rooted connected chord diagram]\label{def:weight}
  For a chord diagram $C\in R_{dec}$ with branch-left vector $\nu(C)$ and decoration $d_i$ for chord $i$ define 
  \begin{align*}
    \omega(C) = \prod_{k=1}^{|C|} { d_ks + \nu_k(C) -2 \choose \nu_k(C) }
  \end{align*}
  where $s$ is the parameter given by our Dyson-Schwinger equation.  
Further, we denote by $\omega_{\ahat{C}}$ the weight associated to $C$ but without the factor corresponding to its base chord $b(C)$.
\end{definition}

Whenever convenient we will write $\omega_C$ instead of $\omega(C)$.


Our first use of the weights will be to show that the renormalization group equation holds for the expansion over chord diagrams which will solve the Dyson-Schwinger equation. 
To keep track of the two expansions -- the expansion of the Dyson-Schwinger equation itself and the chord diagram expansion -- which will ultimately prove to be the same, we will distinguish them by superscripts: $\text{comb}$ for the combinatorial expansion, $\text{dif}$ for the analytic expansion.  Specifically, Let
\[
 G^{\text{dif}}(x,L) = 1 - \sum_{k = 0}^N x^kG^{\text{dif}}(x,\partial_{-\rho = 0})^{1-sk} (e^{-L\rho}-1)F_k(\rho)
\]
and write
\[
 g^{\text{dif}}_{k}(x) =  \frac{(-1)^k}{k!}[L^k]G^{\text{dif}}(x,L) 
\]
On the combinatorial side define
\[
  g^{\text{comb}}_k(x) = \sum_{\substack{C\in \mathcal R^{dec}\\b(C) \geq k}} x^{||C||} \omega_C \ahat{C} a_{d_{b(C)},b(C)-k}
\]
and $G^{\text{comb}}$ as the analogous sum of the $g^{\text{comb}}$.

By Proposition \ref{prop:diff rg eqn} we already know $g^{\text{dif}}$ satisfies the renormalization group equation.  Next we show the analogous result for $g^{\text{comb}}$.
\begin{thm}[Renormalization group equation for $g^{\text{comb}}$]
  \label{thm:rge-equation}
  \begin{align*}
    g^{\text{comb}}_k(x) = g^{\text{comb}}_1(x)\cdot(sx\partial_x - 1)g^{\text{comb}}_{k-1}(x)
  \end{align*}
\end{thm}
To prove this theorem we need to answer the following two questions:
\begin{enumerate}
\item How is the monomial of a chord diagram $C$ recovered from the root share decomposition $C= C_1 \circ_r C_2$?
\item How is the weight of a chord diagram $C$ recovered from the root share decomposition $C= C_1 \circ_r C_2$?
\end{enumerate}
The following two lemmas answer these questions and together they are enough to prove Theorem \ref{thm:rge-equation}.
The monomial associated to a decorated chord diagram $C\in \mathcal R_{\text{dec}}$ with root share decomposition 
$C = C_1 \circ_k C_2$ can be reconstructed from $C_1, C_2$ in the following sense:
\begin{lemma}[RSD monomial Lemma]
  \label{lem:rsd-monomial}
   Let $C_1,C_2\in \mathcal R_{dec}$ with $C = C_1 \circ_k C_2$ and $d, d_{1}, d_2$ the corresponding decorations of the base chords,i.e. $d:= d(b(C)), d_1:=d(b(C_1)), d_2 := d(b(C_2))$, then 
   \begin{align*}
     \ahat{C} a_{d,b(C)-l} = \ahat{C_1} a_{d_1,b(C_1)-1} \ahat{C_2} a_{d_2,b(C_2)-l+1}
   \end{align*}
   where $1 < l < b(C)$
\end{lemma}
\begin{proof}
The proof is the same as the proof of Lemma 4.2 in \cite{MYchord} but keeping track of decorations.
\end{proof}
\begin{example}\label{ex:wheel-example}
 Let $C$ be the wheel with three spokes as a rooted connected chord diagram and choose as decoration a two for the last chord (only the non-trivial decoration is included in the pictures):
\begin{align*}
{\tiny
\begin{tikzpicture}[baseline={([yshift=-2.0ex]current bounding box.center)}]
\draw[black,thin](0,0) circle(0.5);
\foreach \a in {0,60,...,359}
	\draw[fill=black] (\a:0.5) circle(0.3mm);
\draw (0:0.5) circle(0.9mm);
\draw[thin,black](0:0.5) -- (180:0.5);
\draw[thin,black](60:0.5) -- (240:0.5);
;\draw[thin,black](120:0.5) -- (300:0.5);
\draw(120:0.75) node{$2$};\end{tikzpicture}}
\end{align*}
As calculated earlier $C = C_1 \circ_2 C_2$
 \begin{align*}
{\tiny
\begin{tikzpicture}[baseline={([yshift=-2.4ex]current bounding box.center)}]
\draw[black,thin](0,0) circle(0.5);
\foreach \a in {0,60,...,359}
	\draw[fill=black] (\a:0.5) circle(0.3mm);
\draw (0:0.5) circle(0.9mm);
\draw[thin,black](0:0.5) -- (180:0.5);
\draw[thin,black](60:0.5) -- (240:0.5);
;\draw[thin,black](120:0.5) -- (300:0.5);
\draw(120:0.75) node{$2$};\end{tikzpicture}}
= 
   {\tiny
\begin{tikzpicture}[baseline={([yshift=-1.4ex]current bounding box.center)}]
\draw[black,thin](0,0) circle(0.5);
\foreach \a in {0,180,...,359}
	\draw[fill=black] (\a:0.5) circle(0.3mm);
\draw (0:0.5) circle(0.9mm);
\draw[thin,black](0:0.5) -- (180:0.5);
\end{tikzpicture}
} \circ_{2} {\tiny
\begin{tikzpicture}[baseline={([yshift=-2.4ex]current bounding box.center)}]
\draw[black,thin](0,0) circle(0.5);
\foreach \a in {0,90,...,359}
	\draw[fill=black] (\a:0.5) circle(0.3mm);
\draw (0:0.5) circle(0.9mm);
\draw[thin,black](0:0.5) -- (180:0.5);
;\draw[thin,black](90:0.5) -- (270:0.5);
\draw(90:0.75) node{$2$};\end{tikzpicture}
}
 \end{align*}
The left hand side of the previous Lemma is: 
\begin{align*}
\ahat{C} a_{d,b(C)-l} = a_{1,0}^2 a_{2,3-l}
\end{align*}
The right hand side of the previous Lemma contains
\begin{align*}
   \ahat{C_1} = 1\\
   a_{d_1,b(C_1)-l} = a_{1,0}\\
   \ahat{C_2} = a_{1,0}\\
   a_{d_2,b(C_2)-l+1} = a_{2,2-l+1}
\end{align*}
which agrees with the lemma.
\end{example}

Returning to the general case, 
we need to understand how the root share decomposition relates to the branch left vectors.
Consider a chord diagram and its root share decomposition $C = C' \circ_k C''$. We know that branch-left vector $\nu(C')$ is copied into $C$ so only the branch left vector of $C''$ is modified. This yields the following equation:
\begin{align*}
\sum_{k=1}^{2n-1} \omega( C' \circ_k C'' ) =  \omega(C') \sum_{k=1}^{2n-1}\omega(^{\circ_k}C'')
\end{align*}
where $^{\circ_k}$ is defined as follows:
\begin{definition}[Virtual insertion $^{\circ_k}$]
  \label{def:virtual-insertion}
  Let $C\in \mathcal R_{dec}$ then $^{\circ_k} C$ is defined to be the same chord diagram but with modified tree: $T(^{\circ_k} C)$ is $T(C)$ but with an additional vertex $v$ and an additional left child  inserted before the $k$-th vertex $w$. As a result $w$ will be the right child of $v$.
\end{definition}

The following example illustrates this definition.

\begin{example}
  Let $C' =
{\tiny
\begin{tikzpicture}[baseline={([yshift=-2.4ex]current bounding box.center)}]
\draw[black,thin](0,0) circle(0.5);
\foreach \a in {0,90,...,359}
	\draw[fill=black] (\a:0.5) circle(0.3mm);
\draw (270:0.5) circle(0.9mm);
\draw[thin,black](0:0.5) -- (180:0.5);
;\draw[thin,black](90:0.5) -- (270:0.5);
\end{tikzpicture}
}
 $ and $C''  =  \begin{tikzpicture}[baseline={([yshift=-1.4ex]current bounding box.center)}]
\draw[black,thin](0,0) circle(0.5);
\foreach \a in {0,60,...,359}
	\draw[fill=black] (\a:0.5) circle(0.3mm);
\draw (300:0.5) circle(0.9mm);
\draw[thin,black](0:0.5) -- (240:0.5);
\draw[thin,black](60:0.5) -- (180:0.5);
\draw[thin,black](120:0.5) -- (300:0.5);
\end{tikzpicture}
 $ and the associated trees with the induced labeling are: $T_1 := T(C') = \begin{tikzpicture}[baseline={([yshift=0.0ex]current bounding box.center)},level/.style={sibling distance=7mm,level distance=7mm}, level 2/.style={sibling distance=7mm}]
\tikzstyle{every node}=[font=\small]  
  \coordinate
  child{node{$1$}} child{node{$5$}};
\end{tikzpicture}$ and 
$T_2 := T(C'') = \begin{tikzpicture}[baseline={([yshift=0.0ex]current bounding box.center)},level/.style={sibling distance=7mm,level distance=7mm}, level 2/.style={sibling distance=7mm}]
\tikzstyle{every node}=[font=\small]  
  \coordinate
 child{ child{node{$2$}} child{node{$4$}} } child{node{3}};
\end{tikzpicture}$.
\begin{align*}
 &  T_1 \circ_1 T_2 =
\begin{tikzpicture}[baseline={([yshift=0.0ex]current bounding box.center)},level/.style={sibling distance=10mm,level distance=5mm}, level 2/.style={sibling distance=7mm}]
\tikzstyle{every node}=[font=\tiny]  
  \coordinate
child{child{node{$1$}} child{node{$5$}}} child{child{ child{node{$2$}} child{node{$4$}} } child{node{3}}};
\end{tikzpicture}
 \;\;\;\;
  T_1 \circ_2 T_2 = 
\begin{tikzpicture}[baseline={([yshift=0.0ex]current bounding box.center)},level/.style={sibling distance=10mm,level distance=3mm}, level 2/.style={level distance=5mm,sibling distance=20mm}]
\tikzstyle{every node}=[font=\tiny]  
  \coordinate
child{ child{ child{node{1}} child{node{5}} } child{ child{node{2}} child{node{4}} } } child{node{3}};
\end{tikzpicture}
\\
 & T_1 \circ_3 T_2 = 
\begin{tikzpicture}[baseline={([yshift=0.0ex]current bounding box.center)},level/.style={sibling distance=7mm,level distance=5mm}, level 2/.style={sibling distance=7mm}]
\tikzstyle{every node}=[font=\tiny]  
  \coordinate
child{ child{ child{child{node{1}} child{node{5}}} child{node{2}}} child{node{4}}} child{node{3}};
\end{tikzpicture}
 \;\;\;\;
  T_1 \circ_4 T_2 =
\begin{tikzpicture}[baseline={([yshift=0.0ex]current bounding box.center)},level/.style={sibling distance=10mm,level distance=5mm}, level 2/.style={sibling distance=7mm}]
\tikzstyle{every node}=[font=\tiny]  
  \coordinate
child{child{node{2}} child{child{child{node{1}} child{node{5}}} child{node{4}}} } child{node{3}};
\end{tikzpicture}
\\
 & T_1 \circ_5 T_2 = 
\begin{tikzpicture}[baseline={([yshift=0.0ex]current bounding box.center)},level/.style={sibling distance=10mm,level distance=5mm}, level 2/.style={sibling distance=7mm}]
\tikzstyle{every node}=[font=\tiny]  
  \coordinate
child{ child{node{2}} child{node{4}} } child{ child{child{node{1}} child{node{5}} } child{node{3}} };
\end{tikzpicture}
\\
\end{align*}
Note that the branch left vector of $C'$ is never changed, so we can replace it as a marker vertex. This is what the virtual insertion does:
\begin{align*}
  T(^{\circ_2}C'') = \begin{tikzpicture}[baseline={([yshift=0.0ex]current bounding box.center)},level/.style={sibling distance=7mm,level distance=7mm}, level 2/.style={sibling distance=7mm}]
\tikzstyle{every node}=[font=\small]  
  \coordinate
child{ child{} child{ child{node{$2$}} child{node{$4$}} }} child{node{3}};
\end{tikzpicture}
\end{align*}
\end{example}

\begin{lemma}\label{lem:weight-reccurence}
  Let $C',C''$ be decorated chord diagrams where $|C''| = n$, then:
  \begin{align*}
    \sum_{k=1}^{2n-1} \omega( C' \circ_k C'' ) = \omega(C')\omega(C'')( s\lVert C''\rVert -1 )
  \end{align*}
\end{lemma}
\begin{proof}
Note that $C' \circ_k C''$ does not affect the tree form of $C'$ in any way so we get
  \begin{align*}
    \sum_{k=1}^{2n-1} \omega( C' \circ_k C'' ) = \omega(C') \sum_{k=1}^{2n-1} \omega(^{\circ_k}C'')
  \end{align*}
Now notice that there are $\nu_k+1$ possibilities to increase the left branch by 1.
\begin{align*}
  \sum_{k=1}^{2n-1} \omega(^{\circ_k}C'') & = (\nu_1+1) \omega^{+1}(C'') + \ldots + (\nu_n+1) \omega^{+n}(C'')
\end{align*}
where $\omega^{+k}(C'')$ is defined as the weight of $C''$ after incrementing the $k$-th component of the branch left vector:
\begin{align*}
  \omega^{+k}(C'') = \omega(C'')\left(1+\frac{sd_k-2}{\nu_k+1}\right)
\end{align*}
Plugging this into the latter equation, we get the result:
\begin{align*}
  \sum_{k=1}^{2n-1} \omega(^{\circ_k}C'') & = \omega(C'') \sum_{k=1}^n (\nu_k + 1 + sd_k - 2)\\
  & = \omega(C'')\left( n-1 - n + s\sum_{k=1}^n d_k\right) = \omega(C'')(s\lVert C'' \rVert - 1)
\end{align*}
%
\end{proof}
\begin{example}Consider the following decorated chord diagrams (the decorated chords are thickened and the decoration is on the other side from the labeling).
  \begin{align*}
  C' = 
  {\tiny
\begin{tikzpicture}[baseline={([yshift=-2.4ex]current bounding box.center)}]
\draw[black,thin](0,0) circle(0.5);
\foreach \a in {0,90,...,359}
	\draw[fill=black] (\a:0.5) circle(0.3mm);
\draw (0:0.5) circle(0.9mm);
\draw[thin,black](0:0.5) -- (180:0.5);
\draw(0:0.75) node{$1$};\draw[ultra thick,black](90:0.5) -- (270:0.5);
\draw(90:0.75) node{$2$};
\draw(270:0.75) node{$d_1$};
\end{tikzpicture}
} \qquad 
  C'' = 
  {\tiny
\begin{tikzpicture}[baseline={([yshift=-2.4ex]current bounding box.center)}]
\draw[black,thin](0,0) circle(0.5);
\foreach \a in {0,90,...,359}
	\draw[fill=black] (\a:0.5) circle(0.3mm);
\draw (0:0.5) circle(0.9mm);
\draw[thin,black](0:0.5) -- (180:0.5);
\draw(0:0.75) node{$1$};\draw[ultra thick,black](90:0.5) -- (270:0.5);
\draw(90:0.75) node{$2$};
\draw(270:0.75) node{$d_2$};
\end{tikzpicture}
} 
\end{align*}
Clearly, $\nu(C') = \nu(C'') = (0,1)$  and thus we have
\begin{align*}
  & \omega(C') = d_1s-1\qquad \omega(C'') = d_2s-1 \\
  & \omega(C')\omega(C'')(s||C''||-1) = (d_1s-1)(d_2s-1)(s(d_2+1)-1)
\end{align*}
The branch-left vectors for the different insertions are:
\begin{align*}
  & \nu(C' \circ_1 C'') = (0,1,0,2) \qquad  \nu(C' \circ_2 C'') = (0,1,1,1) \qquad \nu(C' \circ_3 C'') = (0,1,0,2)
\end{align*}
Thus, for the the sum of the left hand side of the lemma we have: 
\begin{align*}
  & \omega(C' \circ_1 C'') = \omega(C' \circ_3 C'') = {d_2s \choose 2} (d_{1}s - 1) \\
  & \omega(C' \circ_2 C'') = (d_1s-1)(d_2s-1)(s-1)
\end{align*}
\end{example}
which is what the Lemma~\ref{lem:weight-reccurence} tells us it should be.

\begin{proof}[Proof of theorem~\ref{thm:rge-equation}]:
To prove
  \begin{align*}
    g^{\text{comb}}_k(x) = g^{\text{comb}}_1(x)\cdot(sx\partial_x - 1)g^{\text{comb}}_{k-1}(x),
  \end{align*}
let us do the differential first:
  \begin{align*}
    (sx\partial_x - 1)g^{\text{comb}}_{k-1}(x) = \sum_{\substack{C\in \mathcal R^{dec}\\b(C) \geq k-1}} (s||C||-1)x^{||C||} \omega_C \ahat{C} a_{b(C)-k+1} 
  \end{align*}
  Multiplying $g^{\text{comb}}_1$ from the left we obtain:
  \begin{align*}
    & g^{\text{comb}}_1(x)\cdot (sx\partial_x - 1)g^{\text{comb}}_{k-1}(x) = \\
    & \left[\sum_{\substack{C' \in \mathcal R^{dec}\\b(C') \geq 1}} x^{||C'||} \omega_{C'}  \ahat{C'} a_{b(C')-1}\right]
\left[ \sum_{\substack{C''\in \mathcal R^{dec}\\b(C'') \geq k-1}} (s||C''||-1)x^{||C''||} \omega_{C''} \ahat{C''} a_{b(C'')-k+1} \right] = \\
   & \sum_{\substack{C' \in \mathcal R^{dec}, C''\in \mathcal R^{dec}\\b(C') \geq 1, b(C'')\geq k-1}}  x^{||C'||+||C''||}\omega_{C'}\omega_{C''}(s||C''||-1) \ahat{C'} a_{b(C')-1}  \ahat{C''} a_{b(C'')-k+1}
  \end{align*}
  By Lemma \ref{lem:rsd-monomial} and Lemma \ref{lem:weight-reccurence} the result follows.
\end{proof}

\section{Shuffling trees}\label{sec shuffle}
The next thing we need to understand is, if we have two rooted connected chord diagrams $D_1$ and $D_2$ with corresponding trees $H_1$ and $H_2$, what possible chord diagrams can correspond to the binary rooted tree with left child $H_1$ and right child $H_2$.  The relative order of the labels of $H_1$ and $H_2$ must remain the same, so it is a question of which shuffles of the labels give trees which correspond to chord diagrams.

In the middle of the shuffling process we will need to consider the original labels of $H_1$ and $H_2$ as well as the new labels generated so far.  To keep track of this we sill use the notation $\underline{\mathbb{N}}$ and $\overline{\mathbb{N}}$.  $\underline{\mathbb{N}}$ will hold the labels for the left tree and $\overline{\mathbb{N}}$ will hold the labels for the right tree at the start of the labeling procedure. In the labeling procedure elements of $\underline{\mathbb{N}}$ resp. $\overline{\mathbb{N}}$ will be successively replaced by the final label elements which will be elements of $\mathbb{N}$. 
Let $\underline{<}$ resp. $\overline{<}$ be the strict ordering of $\underline{\mathbb{N}}$ resp. $\overline{\mathbb{N}}$. 
Note that $\underline{<}$ and $\overline{<}$ are not defined to compare an element of $\underline N$ with an element of $\overline N$ and vice versa. However, due to the iterative nature of the algorithm we will consider elements of $\mathbb{N}$ which are assigned so far to be smaller than every element of the remaining labels from $\underline{\mathbb{N}} \cup \overline{\mathbb{N}}$.

Let's define the shuffle product for the special case of two subsets 
$\{\underline 1, \ldots \underline k\}\subset \underline{\mathbb N},\{\overline 1, \ldots , \overline l\} \subset \overline{\mathbb N}$: 

\begin{align*}
\{\underline 1, \ldots \underline k\} \shuffle \{\overline 1, \ldots , \overline l\} := \Bigg\{ (w_1,\ldots,w_{k+l}) :
\begin{tabular}{l}
 $\{w_1,\ldots,w_{k+l}\} = \{ \underline 1, \ldots \underline k\} \cup \{\overline 1, \ldots , \overline l\}$ \\
  and $r < s \Rightarrow w_r \underline{<} w_s \text{ if } w_r,w_s \in \underline{\mathbb{N}}$\\
  and $r < s \Rightarrow w_r \overline{<} w_s \text{ if } w_r,w_s \in \overline{\mathbb{N}}$
 \end{tabular}\Bigg\}  
\end{align*}
\begin{example}
  \begin{align*}
  \{ \underline 1, \underline 2 \} \shuffle \{ \overline 1, \overline 2 \} = \Big\{ (\underline 1, \underline 2, \overline 1, \overline 2), (\underline 1, \overline 1, \underline 2, \overline 2), (\underline 1, \overline 1,\overline 2, \underline 2 ), (\overline 1, \underline 1, \underline 2, \overline 2),(\overline 1,\underline 1,\overline 2, \underline 2),  (\overline 1,\overline 2, \underline 1, \underline 2)
   \Big\}
  \end{align*}
\end{example}

\begin{definition}[Pre-labeling]
  \label{def:pre-labeling}
Let $L \subset \mathbb{N} \cup \underline{\mathbb{N}} \cup \overline{\mathbb{N}}$ be a finite set. We call a bijection
\begin{align*}
  \sigma' : \Lambda(T) \rightarrow L 
\end{align*}
a pre-labeling for $T$ if the image of $\sigma'$ contains elements of $\underline{\mathbb{N}}$ or $\overline{\mathbb{N}}$.  
\end{definition}
In order to use Proposition~\ref{prop:instree-shape} on pre-labeled trees, we will need the notion of smallest removable subtree containing 1 from Definition 4.6 of \cite{MYchord} in a slightly more general setting. First we define what we mean by removing a subtree from a tree.
\begin{definition}[Removing a subtree]
  \label{def:remove-tree}
  Removing a subtree $S\subset T$ rooted at a vertex $w$, denoted by $T\setminus S$, is defined by the following procedure:
  \begin{enumerate}
  \item Every edge and vertex from $S$ will be removed from $T$.
  \item The edge $(w,w')$ where $w\in S$ and $w'\not\in S$ is removed. The edge $(w',w'')$ outgoing from $w'$ where $w''\not\in S$ is contracted.
  \end{enumerate}
\end{definition}
\begin{example}
  Let $T = \begin{tikzpicture}[baseline={([yshift=0.0ex]current bounding box.center)},level/.style={sibling distance=7mm,level distance=7mm}, level 2/.style={sibling distance=7mm}]
\tikzstyle{every node}=[font=\small]  
  \coordinate
   child{child{child{child{node{$1$}}child{node{$5$}} } child{node{$2$}}} child{node{$4$}}} child{node{$3$}};
\end{tikzpicture}$ and $S =\begin{tikzpicture}[baseline={([yshift=0.0ex]current bounding box.center)},level/.style={sibling distance=7mm,level distance=7mm}, level 2/.style={sibling distance=7mm}]
\tikzstyle{every node}=[font=\small]  
  \coordinate
  child{node{$1$}} child{node{$5$}};
\end{tikzpicture}$, then $T\setminus S = \begin{tikzpicture}[baseline={([yshift=0.0ex]current bounding box.center)},level/.style={sibling distance=7mm,level distance=7mm}, level 2/.style={sibling distance=7mm}]
\tikzstyle{every node}=[font=\small]  
  \coordinate
 child{ child{node{$2$}} child{node{$4$}} } child{node{3}};
\end{tikzpicture}$
\end{example}
\begin{definition}[Smallest removable subtree]
  \label{def:smallest-removable-subtree}
  Let $(T,\sigma)$ be a rooted, plane, leaf labeled, binary tree $T$ with a (pre-)labeling $\sigma$.
  A smallest removable subtree $S$ of $T$ is defined to be the smallest tree such that $T\setminus S$ maintains P1 of Proposition~\ref{prop:instree-shape}. 
\end{definition}

In \cite{MYchord} a complete characterization of labelings of binary trees coming from chord diagrams is given.  Specifically:
\begin{prop}
  \label{prop:instree-shape}
Let $\mathcal T_n$ be the set of rooted, plane, leaf labeled, binary trees with $n$ leaves such that for every $(T,\sigma)\in \mathcal T_n$ the following two properties hold
\begin{enumerate}
\item [P1)] At any vertex $v$ that is not a leaf the smallest label in the left subtree of $v$ is smaller than the label at the end  of the fully right branch of the right subtree.
\item [P2)] Let $H$ be the smallest removable subtree of T containing 1. $H$ contains exactly the following leaf labels: 
  \begin{align*}
   \text{Im}(\sigma|_H)  = \Big\{  1,l(T)-l(H)+2, l(T)-l(H)+3, \ldots, l(T) \Big\}
  \end{align*}
  where $l(\cdot)$ denotes here the maximal label of a tree. Note that $H$ is the left side of the root share decomposition of trees.

\end{enumerate}  
Furthermore, P1 and P2 must stay true recursively in the following sense. Let $T = H \circ_r (T\setminus H)$ for some $r$ then P1 and P2 must hold for $T\setminus H$.
Then every $(T,\sigma)\in \mathcal T_n$ represents a unique rooted connected chord diagram of size $n$, so 
\begin{align*}
  \mathcal T_n = \{ T(C) : C\in \mathcal{R}, |C| = n \}
\end{align*}
\end{prop}
\begin{proof}See \cite{MYchord}
\end{proof}
Call a labeling satisfying these properties admissible.

Let $k = 1,\ldots,|D_1|$ and $m := b(D_2)$ the base chord corresponding to $D_2 = \mathcal T^{-1}(H_2)$, $\underline 1, \ldots \underline n \in \underline{\mathbb N}$ the pre-labeling for $H_1$ and $\overline 1,\ldots \overline h \in \overline{\mathbb N}$ the pre-labeling for $H_2$. The following procedure assigns to a shuffle 
\begin{align*}
  w=(w_1\ldots w_{k+m}) \in \{\underline 1, \ldots \underline k\} \shuffle \{\overline 1, \ldots , \overline m\}
\end{align*}
an admissible labeling $\sigma = \sigma(w)$ for the tree with $H_1$ and $H_2$ the children of the root. Because this tree will be well defined, we call it $H_1 \diamond_{\sigma} H_2$, see Definition~\ref{def:diamond-operation}. The set of shuffles $\{\underline 1, \ldots \underline k\} \shuffle \{\overline 1, \ldots , \overline m\}$ is therefore the set of admissible shuffles associated to $D_1,D_2$ or equivalently to $H_1,H_2$ and will be denoted by $D_1 \shuffle D_2$ resp. $H_1 \shuffle H_2$.
\begin{prop}
  \label{prop:diamond-leaf-labeling}
  Let $w$ be a shuffle of $\{\underline 1, \ldots \underline k\}$ and $\{\overline 1, \ldots , \overline m\}$ and $(H_1,\sigma_{H_1}), (H_2,\sigma_{H_2})$ as before, then the following algorithm produces an admissible labeling $\sigma$ and a unique leaf labeled tree $(T,\sigma) \in \mathcal T$:
\begin{enumerate}
\item Graft the left and right tree $H_1$ and $H_2$ at a new root, merge the pre-labelings and call this tree $(T_1,\sigma_1)$. To be more specific: $T_1 = B_+(H_1H_2)$ and 
  \begin{align*}
\sigma_1: \Lambda(H_{1}) \cup \Lambda(H_2) \rightarrow \{\underline 1,\ldots,\underline n\} \cup 
\{\overline 1,\ldots,\overline h\} 
  \end{align*}
is given by 
\begin{align*}
  \sigma_1(\lambda) = \begin{cases}\sigma_{H_1}(\lambda) & \text{ if } \lambda\in H_1 \\ \sigma_{H_2}(\lambda) & \text{otherwise} \end{cases}
\end{align*}
\item For each $l=1,\ldots,k+m$, replace the pre-label $\omega_l$ by the label $l \in \mathbb{N}$, i.e. modify $\sigma_1$ such that $\sigma_1(w_l) = l$
\item Assign the label $b(D_2)+k \in \mathbb{N}$ to the fully right branch leaf of $T_1$,i.e. modify $\sigma_1$ such that
  \begin{align*}
    \sigma_1 (\lambda) = b(D_2)+k
  \end{align*}
where $\lambda$ is the leaf of the fully right branch of $T_1$.
\newpage
\item Apply \textbf{LABEL}($T_1,\sigma_1,b(D_2)+k+1)$. The labeling procedure \textbf{LABEL} is defined as follows:
\begin{lstlisting}[numbers=none]
  LABEL(T,ref $\sigma$,ref l) {
    if $\sigma$ is an admissible label {
      return (T,$\sigma$)
    } 
    s := 0
    if $Im(\sigma) \subset \mathbb{N}\cup \underline{\mathbb{N}}$ or $Im(\sigma) \subset \mathbb{N}\cup \overline{\mathbb{N}}$ {
      // replace the pre-label elements by the next labels $l$ 
      // in the order that is induced by $\underline{<}$ resp. $\overline{<}$
      if $Im(\sigma) \subset \mathbb{N}\cup \underline{\mathbb{N}}$ {
        s := $|\underline{\mathbb{N}} \cap Im(\sigma)|$
        Let $\{ \underline{\lambda}_1 \:\underline{<} \ldots \underline{<} \:\underline{\lambda}_s \} = \underline{\mathbb{N}} \cap Im(\sigma)$
        for $i=1\ldots s$ {
          replace pre-label $\underline{\lambda}_i$ by $l+i$
        }
      } 
      if $Im(\sigma) \subset \mathbb{N}\cup \overline{\mathbb{N}}$ {
        s := $|\overline{\mathbb{N}} \cap Im(\sigma)|$
        Let $\{ \overline{\lambda}_1 \:\overline{<} \ldots \overline{<} \:\overline{\lambda}_s \} = \overline{\mathbb{N}} \cap Im(\sigma)$
        for $i=1\ldots s$ {
          replace pre-label $\overline{\lambda}_i$ by $l+i$
        }
      }
    } 
    // using the extended definition
    // of smallest removable subtree to get 
    // the root share decomposition on
    // the level of trees
    $T = T' \circ_r T''$ 
    LABEL($T''$,$\sigma$,l+s)
    LABEL($T'$,$\sigma$,l+s)   
  }
\end{lstlisting}
Note that \textbf{LABEL} does not change the form of $T$.
\end{enumerate}  
\end{prop}
\begin{proof}
  This is the content of Lemma 4.12 of \cite{MYchord} with the shuffling algorithm clarified.  The proof is the same.
\end{proof}

\section{Diamond decomposition}
\label{sec:diamond}

Now let us consider the decomposition mentioned in the previous section in more detail.
When starting with a tree $T(C)$ associated to a chord diagram, by removing the root we obtain a left and a right tree that define chord diagrams by themselves. Let us call those chord diagrams $D_1$ and $D_2$, They are well defined for every chord diagram $C$ and so we are able to define the diamond operation on chord diagrams. This operation, which will be defined in detail in Definition \ref{def:diamond-operation}, is needed for some technical lemmas that we need to prove the main theorem. The lemmas roughly say that summing over a set of chord diagrams of fixed size $n$ is the same as summing over all possible decompositions of $C$ into $D_1, D_2$. To be more specific we will need to prove: 

  \begin{align}\label{eq:tree-binomial}
    \sum_{\substack{||C||=i+1\\b(C)=j+1}} \omega_{\hat C} \ahat{C}
     = \sum_{k=1}^i\sum_{l=1}^j {j \choose l} \left(
      \sum_{\substack{||D_1||=k\\b(D_1)\geq l}} \omega_{D_1}\ahat{D_1} a_{b(D_1)-l}\right)
    \left( \sum_{\substack{||D_2||=i-k+1\\b(D_2)= j-l+1}}
      \omega_{\hat D_2}\ahat{D_2}\right) 
  \end{align}

The proof relies crucially on the Proposition~\ref{prop:diamond-leaf-labeling}, which tells us in how many ways two given trees can be grafted together.

If we decompose trees or chord diagrams by their left and right subtree, this is a well defined operation. However, if we start with two trees it is not clear which labeling the diamond operation should give and Proposition \ref{prop:diamond-leaf-labeling} tells us what labelings are possible for it. This being said, we define the diamond operation in the following way

\begin{definition}[Diamond operation on trees and chord diagrams]
  \label{def:diamond-operation}
  Let $T_1, T_2 \in T(\mathcal R)$, $\lambda$ be a leaf labeling of size $l(T_1) + l(T_2)$ where $l(\cdot)$ denotes the numbers of leaves, then we define $T_1 \diamond_{\lambda} T_2$  to be the unique tree that has $T_1$ as left tree, $T_2$ as right tree and $\lambda$ as leaf labeling.
  If we take the induced labeling of a tree $T$, we write $T_1 \diamond_{T} T_2$.
  Analogously, we write for chord diagrams $D_1$, $D_2$ and a chord labeling $\mu$ of size $|D_1| + |D_2|$ $D_1 \diamond_{\mu} D_2$  and for the induced labeling of a chord diagram $C$, we write $D_1 \diamond_{C} D_2$. In the case no labeling is assigned, i.e. $T_1\diamond T_2$ resp. $D_1\diamond D_2$ is defined to be the set of all possible labelings.
\end{definition}

\begin{example}
Consider 
$
C = 
\begin{tikzpicture}[baseline={([yshift=+0.5ex]current bounding box.center)}]\tikzstyle{every node}=[font=\tiny]
\draw[black,thin](0,0) circle(0.5);
\draw[fill=black](270:0.5) circle(0.5mm); 
\draw(270:0.5) circle(1.0mm);
\draw[fill=black](0:0.5) circle(0.5mm); 
\draw[fill=black](90:0.5)  circle(0.5mm); 
\draw[fill=black](180:0.5)  circle(0.5mm); 
\draw[thin,black](270:0.5) -- (90:0.5); 
\draw[thin,black](0:0.5) -- (180:0.5); 
\draw(270:0.8) node{1};
\draw(180:0.8) node{(4)};
\end{tikzpicture}  
\;\circ_k\;
\begin{tikzpicture}[baseline={([yshift=+0.5ex]current bounding box.center)}]\tikzstyle{every node}=[font=\tiny]
\draw[black,thin](0,0) circle(0.5);
\draw[fill=black](270:0.5) circle(0.5mm); 
\draw(270:0.5) circle(1.0mm);
\draw[fill=black](0:0.5) circle(0.5mm); 
\draw[fill=black](90:0.5)  circle(0.5mm); 
\draw[fill=black](180:0.5)  circle(0.5mm); 
\draw[thin,black](270:0.5) -- (90:0.5); 
\draw[thin,black](0:0.5) -- (180:0.5); 
\draw(270:0.8) node{(2)};
\draw(0:0.8) node{(3)};
\end{tikzpicture}  $
Then depending on $k$ we have the following $C$ and $C = D_1\diamond_C D_2$
\begin{center}
\begin{tabular}{|l|c|c|c|}
\hline
 k  &  1  &  2  &  3  \\
\hline
 C   & \begin{tikzpicture}[baseline={([yshift=+2.0ex]current bounding box.center)}]\tikzstyle{every node}=[font=\small]
\draw[black,thin](0,0) circle(1);
\draw[fill=black](270:1.0) circle(0.5mm); 
\draw(270:1.0) circle(1.0mm);
\draw[fill=black](0:1.0) circle(0.5mm); 
\draw[fill=black](90:1.0)  circle(0.5mm); 
\draw[fill=black](180:1.0)  circle(0.5mm); 
\draw[fill=black](215:1.0)  circle(0.5mm); 
\draw[fill=black](325:1.0)  circle(0.5mm); 
\draw[fill=black](300:1.0)  circle(0.5mm); 
\draw[fill=black](345:1.0)  circle(0.5mm); 
\draw[thin,black](270:1.0) -- (90:1.0); 
\draw[thin,black](0:1.0) -- (180:1.0); 
\draw[thin,black](215:1.0) -- (325:1.0); 
\draw plot [smooth] coordinates {(300:1.0) (320:0.7) (345:1.0)};
\draw(180:1.3) node{4};
\draw(270:1.3) node{1};
\draw(300:1.3) node{2};
\draw(210:1.3) node{3};
\end{tikzpicture}      &  \begin{tikzpicture}[baseline={([yshift=+2.0ex]current bounding box.center)}]\tikzstyle{every node}=[font=\small]
\draw[black,thin](0,0) circle(1);
\draw[fill=black](270:1.0) circle(0.5mm); 
\draw(215:1.0) circle(1.0mm);
\draw[fill=black](0:1.0) circle(0.5mm); 
\draw[fill=black](90:1.0)  circle(0.5mm); 
\draw[fill=black](180:1.0)  circle(0.5mm); 
\draw[fill=black](215:1.0)  circle(0.5mm); 
\draw[fill=black](45:1.0)  circle(0.5mm); 
\draw[fill=black](25:1.0)  circle(0.5mm); 
\draw[fill=black](65:1.0)  circle(0.5mm); 
\draw[thin,black](270:1.0) -- (90:1.0); 
\draw[thin,black](0:1.0) -- (180:1.0); 
\draw[thin,black](215:1.0) -- (45:1.0); 
\draw plot [smooth] coordinates {(25:1.0) (45:0.7) (65:1.0)};
\draw(215:1.3) node{1};
\draw(270:1.3) node{2};
\draw(0:1.3) node{3};
\draw(20:1.3) node{4};
\end{tikzpicture}  
   & \begin{tikzpicture}[baseline={([yshift=+2.0ex]current bounding box.center)}]\tikzstyle{every node}=[font=\small]
\draw[black,thin](0,0) circle(1);
\draw[fill=black](270:1.0) circle(0.5mm); 
\draw(225:1.0) circle(1.0mm);
\draw[fill=black](0:1.0) circle(0.5mm); 
\draw[fill=black](90:1.0)  circle(0.5mm); 
\draw[fill=black](180:1.0)  circle(0.5mm); 
\draw[fill=black](225:1.0)  circle(0.5mm); 
\draw[fill=black](115:1.0)  circle(0.5mm); 
\draw[fill=black](135:1.0)  circle(0.5mm); 
\draw[fill=black](155:1.0)  circle(0.5mm); 
\draw[thin,black](270:1.0) -- (90:1.0); 
\draw[thin,black](0:1.0) -- (180:1.0); 
\draw[thin,black](225:1.0) -- (135:1.0); 
\draw plot [smooth] coordinates {(115:1.0) (135:0.7) (155:1.0)};
\draw(235:1.3) node{1};
\draw(270:1.3) node{2};
\draw(0:1.3) node{3};
\draw(160:1.3) node{4};
\end{tikzpicture}      \\
 $T=H_1\diamond_T H_2$   &
  \begin{tikzpicture}[baseline={([yshift=-1.4ex]current bounding box.center)},level/.style={sibling distance=10mm,level distance=5mm}, level 2/.style={sibling distance=5mm}]\tikzstyle{every node}=[font=\tiny]
\coordinate child {
        child {node {1}}
        child {node {4}}       
      }
      child {
        child {node {2}}
        child {node {3}}
      };
  \end{tikzpicture}
    &   \begin{tikzpicture}[baseline={([yshift=-1.4ex]current bounding box.center)},level/.style={sibling distance=7mm,level distance=5mm}, level 2/.style={sibling distance=5mm}]\tikzstyle{every node}=[font=\tiny]
\coordinate child {
        child {
        child {node {1}}
        child {node {4}}
        }
        child {node {2}}       
      }
      child {node {3}};
  \end{tikzpicture}     &    \begin{tikzpicture}[baseline={([yshift=-1.4ex]current bounding box.center)},level/.style={sibling distance=5mm,level distance=5mm}, level 2/.style={sibling distance=5mm}]\tikzstyle{every node}=[font=\tiny]\coordinate
      child {node {2}}
      child {
        child {
        child {node {1}}
        child {node {4}}
        }
        child {node {3}}       
      };
  \end{tikzpicture}   \\
\hline
 $H_1$ & \begin{tikzpicture}[baseline={([yshift=+.5ex]current bounding box.center)},level/.style={sibling distance=5mm,level distance=5mm}, level 2/.style={sibling distance=5mm}]\tikzstyle{every node}=[font=\tiny]
    \coordinate child {node {1}} child {node {4}};
  \end{tikzpicture} &  \begin{tikzpicture}[baseline={([yshift=-1.4ex]current bounding box.center)},level/.style={sibling distance=5mm,level distance=5mm}, level 2/.style={sibling distance=5mm}]\tikzstyle{every node}=[font=\tiny]\coordinate
        child { child {node {1}} child {node {4}} }  child {node {2}};
  \end{tikzpicture}  &  \begin{tikzpicture}[baseline={([yshift=-1.4ex]current bounding box.center)},level/.style={sibling distance=5mm,level distance=5mm}, level 2/.style={sibling distance=5mm}]\tikzstyle{every node}=[font=\tiny]\coordinate
      child {node {2}};
  \end{tikzpicture} \\
 $H_2$ & \begin{tikzpicture}[baseline={([yshift=+.5ex]current bounding box.center)},level/.style={sibling distance=5mm,level distance=5mm}, level 2/.style={sibling distance=5mm}]\tikzstyle{every node}=[font=\tiny]
    \coordinate child {node {2}} child {node {3}};
  \end{tikzpicture} &\begin{tikzpicture}[baseline={([yshift=-1.4ex]current bounding box.center)},level/.style={sibling distance=5mm,level distance=5mm}, level 2/.style={sibling distance=5mm}]\tikzstyle{every node}=[font=\tiny]\coordinate
      child {node {3}};
  \end{tikzpicture} & \begin{tikzpicture}[baseline={([yshift=-1.4ex]current bounding box.center)},level/.style={sibling distance=5mm,level distance=5mm}, level 2/.style={sibling distance=5mm}]\tikzstyle{every node}=[font=\tiny]\coordinate
        child { child {node {1}} child {node {4}} }  child {node {3}};
  \end{tikzpicture} \\
 $D_1$  & 
\begin{tikzpicture}[baseline={([yshift=+0.5ex]current bounding box.center)}]\tikzstyle{every node}=[font=\tiny]
\draw[black,thin](0,0) circle(.5);
\draw[fill=black](270:0.5) circle(0.5mm); 
\draw(270:0.5) circle(1.0mm);
\draw[fill=black](0:0.5) circle(0.5mm); 
\draw[fill=black](90:0.5)  circle(0.5mm); 
\draw[fill=black](180:0.5)  circle(0.5mm); 
\draw[thin,black](270:0.5) -- (90:0.5); 
\draw[thin,black](0:0.5) -- (180:0.5); 
\draw(270:0.8) node{1};
\draw(180:0.8) node{2};
\end{tikzpicture}     & 
\begin{tikzpicture}[baseline={([yshift=+0.5ex]current bounding box.center)}]\tikzstyle{every node}=[font=\tiny]
\draw[black,thin](0,0) circle(.5);
\draw[fill=black](270:0.5) circle(0.5mm); 
\draw(270:0.5) circle(1.0mm);
\draw[fill=black](30:0.5) circle(0.5mm); 
\draw[fill=black](90:0.5)  circle(0.5mm); 
\draw[fill=black](150:0.5)  circle(0.5mm); 
\draw[fill=black](330:0.5)  circle(0.5mm); 
\draw[fill=black](210:0.5)  circle(0.5mm); 
\draw[thin,black](270:0.5) -- (90:0.5); 
\draw[thin,black](30:0.5) -- (150:0.5); 
\draw[thin,black](330:0.5) -- (210:0.5); 
\draw(270:0.8) node{1};
\draw(200:0.8) node{2};
\draw(150:0.8) node{3};
\end{tikzpicture}
    &
\begin{tikzpicture}[baseline={([yshift=+0.5ex]current bounding box.center)}]\tikzstyle{every node}=[font=\tiny]
\draw[black,thin](0,0) circle(.5);
\draw[fill=black](270:0.5) circle(0.5mm); 
\draw(270:0.5) circle(1.0mm);
\draw[fill=black](90:0.5)  circle(0.5mm); 
\draw[thin,black](270:0.5) -- (90:0.5); 
\draw(270:0.8) node{1};
\end{tikzpicture}
     \\
 $D_2$  & \begin{tikzpicture}[baseline={([yshift=+0.5ex]current bounding box.center)}]\tikzstyle{every node}=[font=\tiny]
\draw[black,thin](0,0) circle(.5);
\draw[fill=black](270:0.5) circle(0.5mm); 
\draw(270:0.5) circle(1.0mm);
\draw[fill=black](0:0.5) circle(0.5mm); 
\draw[fill=black](90:0.5)  circle(0.5mm); 
\draw[fill=black](180:0.5)  circle(0.5mm); 
\draw[thin,black](270:0.5) -- (90:0.5); 
\draw[thin,black](0:0.5) -- (180:0.5); 
\draw(270:0.8) node{1};
\draw(180:0.8) node{2};
\end{tikzpicture}    & 
\begin{tikzpicture}[baseline={([yshift=+0.5ex]current bounding box.center)}]\tikzstyle{every node}=[font=\tiny]
\draw[black,thin](0,0) circle(.5);
\draw[fill=black](270:0.5) circle(0.5mm); 
\draw(270:0.5) circle(1.0mm);
\draw[fill=black](90:0.5)  circle(0.5mm); 
\draw[thin,black](270:0.5) -- (90:0.5); 
\draw(270:0.8) node{1};
\end{tikzpicture}
    & 
\begin{tikzpicture}[baseline={([yshift=+0.5ex]current bounding box.center)}]\tikzstyle{every node}=[font=\tiny]
\draw[black,thin](0,0) circle(.5);
\draw[fill=black](270:0.5) circle(0.5mm); 
\draw(270:0.5) circle(1.0mm);
\draw[fill=black](30:0.5) circle(0.5mm); 
\draw[fill=black](90:0.5)  circle(0.5mm); 
\draw[fill=black](150:0.5)  circle(0.5mm); 
\draw[fill=black](330:0.5)  circle(0.5mm); 
\draw[fill=black](210:0.5)  circle(0.5mm); 
\draw[thin,black](270:0.5) -- (90:0.5); 
\draw[thin,black](30:0.5) -- (150:0.5); 
\draw[thin,black](330:0.5) -- (210:0.5); 
\draw(270:0.8) node{1};
\draw(200:0.8) node{2};
\draw(150:0.8) node{3};
\end{tikzpicture}
    \\
\hline
\end{tabular}
\end{center}
\end{example}

The interplay between the diamond operation and the root share decomposition will be an important tool for proofs.  The form of their relationship depends on the insertion place of the root share decomposition as stated below:

\begin{prop}\label{prop:threeway-assoc}
Let $C$ be a chord diagram with $|C|\geq 3$ and $C = C' \circ_k C_2$ where $C_2 = C'' \diamond_{C_2} C'''$, then
\begin{align*}
  C'\circ_k (C'' \diamond_{C_2} C''') = 
\begin{cases}
  (C'\circ_{k-1} C'')\diamond_C C''' & |C''| \leq k-1, k>1\\
  C'\diamond_C (C''\diamond_{C_2} C''')  & k=1\\
  C''\diamond_C (C' \circ_{k-|C''|-1} C''') & else
\end{cases}
\end{align*}  
\end{prop}
\begin{proof}
  Let $|C| \geq 3$ with $C = C' \circ_{k} C_2$. There are three cases to consider if we look at $T = T(C)$:
  \begin{enumerate}
  \item $k$ is the root of $T$: This is the case $k=1$ and so root share decomposition and diamond decomposition coincide.
  \item $k$ lies in the left subtree of $T$: Let $T(D_1)$ be the left subtree of $T$ and $D'_1$ the diagram corresponding chord diagram to left subtree of $T(D_1)$, then 
    \begin{align*}
      D_1 = C' \circ_{k-1} D'_1
    \end{align*}
  \item $k$ lies in the right subtree of $T$: Let $T(D_2)$ be the right subtree of $T$ and $D'_2$ the diagram corresponding chord diagram to left subtree of $T(D_2)$, then
    \begin{align*}
      D_2 = C' \circ_{k-1-|D_1|} D'_2
    \end{align*}
  \end{enumerate}
\end{proof}

We have some control over the base chord under the diamond decomposition as seen in the following lemma which is Lemma 4.1 of \cite{MYchord}.
\begin{lemma}[Triangle inequality for the base chords]
  \label{lem:triangleineq}
  \begin{align*}
    b(D_1 \diamond D_2) \leq b(D_1) + b(D_2)
  \end{align*}
\end{lemma}

Again we see the interplay between the strengths and weaknesses of $\circ$ and $\diamond$.
The terminal sets are under control for $\circ_k$ but we don't know
what they do for $\diamond$. The branch-left vectors are under control
for $\diamond$ (let $C= D_1\diamond D_2$, then the only component that is increased is the base chord
of $D_2$) but we know only partial results on $\circ_k$. 
To prove equation \ref{eq:tree-binomial} we need the following Lemma which explains how the weighted monomials of chord diagrams behave under the diamond operation.
\begin{lemma}
  \label{lem:declemma411}
  Let $C\in \mathcal R$ with $|C| \geq 2$ and $T,H_1,H_2,D_1,D_2$ as before.
  Let $d = d_{b(D_1)}$ be the decoration of the smallest terminal chord of $D_1$, then
  \begin{align*}
    & \omega_{\hat C} \ahat{C} = \omega_{D_1} \omega_{\hat
      D_2} \ahat{D_1} \ahat{D_2} a_{d,b(D_1)+b(D_2)-b(C)} \\
  \end{align*}
\end{lemma}
\begin{proof}
  The Lemma follows from the following two claims:
  \begin{enumerate}
  \item[Claim 1:] $\omega_{\hat C} =\omega_{D_1} \omega_{\hat D_2} $. Remember that
      $\omega_{\hat C} = \prod_{k\neq b(C)} {sd_k+\nu_k-2 \choose \nu_k}$ and notice that in the induced labeling the base chord of
$D_2$ and $C$ are the same, so we have $b(D_2) = b(C)$ because it is the fully right branch leaf of the tree that correspond to $D_2$ as well that of $C$, hence the product of chords is the same on both sides.
  \item[Claim 2:] $\ahat{C} = \ahat{D_1} \ahat{D_2} a_{d,b(D_1)+b(D_2)-b(C)}$. This is Lemma 4.11 of \cite{MYchord}.
  \end{enumerate}
\end{proof}

\begin{example}
 Consider the following chord diagram with arbitrary decorations $d_1,\ldots,d_4$ and arbitrary $s\neq 1$:
 \begin{align*}
   C = \begin{tikzpicture}[baseline={([yshift=+1.0ex]current bounding box.center)}]\tikzstyle{every node}=[font=\small]
\draw[black,thin](0,0) circle(1);
\draw[fill=black](270:1.0) circle(0.5mm); 
\draw(225:1.0) circle(1.0mm);
\draw[fill=black](0:1.0) circle(0.5mm); 
\draw[fill=black](90:1.0)  circle(0.5mm); 
\draw[fill=black](180:1.0)  circle(0.5mm); 
\draw[fill=black](225:1.0)  circle(0.5mm); 
\draw[fill=black](115:1.0)  circle(0.5mm); 
\draw[fill=black](135:1.0)  circle(0.5mm); 
\draw[fill=black](155:1.0)  circle(0.5mm); 
\draw[thin,black](270:1.0) -- (90:1.0); 
\draw[thin,black](0:1.0) -- (180:1.0); 
\draw[thin,black](225:1.0) -- (135:1.0); 
\draw plot [smooth] coordinates {(115:1.0) (135:0.7) (155:1.0)};
\draw(235:1.3) node{$d_1$};
\draw(270:1.3) node{$d_2$};
\draw(0:1.3) node{$d_3$};
\draw(160:1.3) node{$d_4$};
\end{tikzpicture}
 \end{align*}
It has terminals $\text{ter}(C) = \{3,4\}$, so $b(C) = 3$.
The corresponding tree is 
\begin{align*}
T(C) = \begin{tikzpicture}[baseline={([yshift=-0.5ex]current bounding box.center)},level/.style={sibling distance=5mm,level distance=5mm}, level 2/.style={sibling distance=5mm}]\tikzstyle{every node}=[font=\tiny]\coordinate
      child {node {2}}
      child {
        child {
        child {node {1}}
        child {node {4}}
        }
        child {node {3}}       
      };
  \end{tikzpicture}  
\end{align*}
so all in all we have for the left hand side of the previous lemma:
\begin{align*}
 \omega_{\hat C} \ahat{C} = (d_4s-1)a_{d_4,1}a_{d_1,0}a_{d_2,0}   
\end{align*}
For the right hand side we have the following trees and diagrams (the decoration is inherited but the labeling is normalized): 
\begin{align*}
  D_1 = \begin{tikzpicture}[baseline={([yshift=+0.5ex]current bounding box.center)}]\tikzstyle{every node}=[font=\tiny]
\draw[black,thin](0,0) circle(.5);
\draw[fill=black](270:0.5) circle(0.5mm); 
\draw(270:0.5) circle(1.0mm);
\draw[fill=black](90:0.5)  circle(0.5mm); 
\draw[thin,black](270:0.5) -- (90:0.5); 
\draw(270:0.8) node{1, $d_2$};
\end{tikzpicture}\qquad
  D_2 = \begin{tikzpicture}[baseline={([yshift=+0.5ex]current bounding box.center)}]\tikzstyle{every node}=[font=\tiny]
\draw[black,thin](0,0) circle(.5);
\draw[fill=black](270:0.5) circle(0.5mm); 
\draw(270:0.5) circle(1.0mm);
\draw[fill=black](30:0.5) circle(0.5mm); 
\draw[fill=black](90:0.5)  circle(0.5mm); 
\draw[fill=black](150:0.5)  circle(0.5mm); 
\draw[fill=black](330:0.5)  circle(0.5mm); 
\draw[fill=black](210:0.5)  circle(0.5mm); 
\draw[thin,black](270:0.5) -- (90:0.5); 
\draw[thin,black](30:0.5) -- (150:0.5); 
\draw[thin,black](330:0.5) -- (210:0.5); 
\draw(270:0.8) node{1, $d_1$};
\draw(200:0.9) node{2, $d_3$};
\draw(150:0.9) node{3, $d_4$};
\end{tikzpicture}\qquad 
H_1 =  \begin{tikzpicture}[baseline={([yshift=0.0ex]current bounding box.center)},level/.style={sibling distance=5mm,level distance=5mm}, level 2/.style={sibling distance=5mm}]\tikzstyle{every node}=[font=\tiny]\coordinate
      child {node {1}};
  \end{tikzpicture} \qquad H_2 = \begin{tikzpicture}[baseline={([yshift=0.0ex]current bounding box.center)},level/.style={sibling distance=5mm,level distance=5mm}, level 2/.style={sibling distance=5mm}]\tikzstyle{every node}=[font=\tiny]\coordinate
        child { child {node {1}} child {node {3}} }  child {node {2}};
  \end{tikzpicture}
\end{align*}
So we have for the right hand side:
\begin{align*}
&  \ahat{D_1} = 1\\ & \ahat{D_2} = a_{d_4,1}a_{d_1,0}\\ & \omega_{D_1} = 1\\ & \omega_{\hat D_2} = d_4s-1 \\
& a_{d(b(D_1)),b(D_1)+b(D_2)-b(C)} = a_{d_2,0}
\end{align*}
multiplying this we indeed get the same as the left hand side as stated in Lemma \ref{lem:declemma411}.
\end{example}

\begin{lemma}\label{prop:big-restricted-sum-eq}
For $j \in \mathbb Z_{\geq 0}$ and every $k\in \mathbb N$
\begin{align*}
  \sum_{\substack{||C|| = i+1\\d_{j+1} = 1\\\nu_{j+1} = n\\b(C) = j+1} } \ahat{C} \omega_{\hat C} =  \sum_{\substack{||C|| = i+k\\d_{j+1} = k\\\nu_{j+1} = n\\b(C) = j+1} } \ahat{C} \omega_{\hat C} 
\end{align*}
\end{lemma}
\begin{proof}
Let
\begin{align*}
  \mathcal C_{i,j,k} := \left\{C\in \mathcal{R}^{\text{dec}}: ||C||= i+k, d_{j+1} = k, \nu_{j+1} = n, b(C) = j+1 \right\}
\end{align*}
We need to show that 
\begin{align*}
  A_{i,j} := \left\{ \ahat{C} \omega_{\hat C}: C\in \mathcal C_{i,j,1} \right\} \text{ is in bijection with every } A_{i,j,k} := \left\{ \ahat{C} \omega_{\hat C}: C\in \mathcal C_{i,j,k} \right\}
\end{align*}
For a fixed $k$ the map that replaces the decoration $d_{j+1} = 1$ by $k$ clearly defines a bijection between $\mathcal C_{i,j,1}$ and $\mathcal C_{i,j,k}$. This map lifts to a bijection $A_{i,j} \leftrightarrow A_{i,j,k}$ because the decoration $d_{j+1}$ is ignored by definition of $\ahat{C}$ and $\omega_{\hat C}$, since $b(C) = j+1$.
\end{proof}

\begin{prop}[Decorated version of Proposition 4.3 from \cite{MYchord}]
  \label{prop:decprop43}
  \begin{align*}
    \sum_{\substack{||C||=i+1\\b(C)=j+1}} \omega_{\hat C} \ahat{C} = \sum_{k=1}^i\sum_{l=1}^j {j \choose l} \left(
      \sum_{\substack{||D_1||=k\\b(D_1)\geq l}} \omega_{D_1}\ahat{D_1} a_{b(D_1)-l}\right)
    \left( \sum_{\substack{||D_2||=i-k+1\\b(D_2)= j-l+1}}
      \omega_{\hat D_2}\ahat{D_2}\right) 
  \end{align*}
\end{prop}

\begin{proof}
We know that each chord diagram $C$ of size $i+1$ can be tree
decomposed to $C = D_1 \diamond_C D_2$ and that in this case $b(D_1) +
b(D_2) \geq b(C)$ by the triangle inequality \ref{lem:triangleineq}. 
However, given $b(D_1) \geq l$ for fixed $l$ and
$b(D_2) = j -l + 1$ for fixed $j$ there are ${j \choose l}$
possibilities for $D_1, D_2$ such that $C = D_1 \diamond_C D_2$ by Proposition \ref{prop:diamond-leaf-labeling}.
Furthermore in all cases $b(C)= j+1$.
Therefore the sum on the left hand side of the statement splits as follows:
\begin{align*}
    \sum_{\substack{||C||=i+1\\b(C)=j+1}} = \sum_{k=1}^i\sum_{l=1}^j {j \choose l} \left(
      \sum_{\substack{||D_1||=k\\b(D_1)\geq l}} \right)
    \left( \sum_{\substack{||D_2||=i-k+1\\b(D_2)= j-l+1}} \right) 
\end{align*}
Now given a monomial $\omega_{\hat C} \ahat{C}$ we know how to decompose it into the monomials as needed by
Lemma \ref{lem:declemma411}. Inserting them into the sums proves the proposition.
\end{proof}
\begin{example}
  Let $i = 3, j = 1$ and $N \geq 3$ and $s\in \mathbb Z_{\geq 2}$ be arbitrary. We have to consider all chord diagrams with $||C|| = 4, b(C) = 2$. For the decorations we need to consider all compositions of four:
  \begin{align*}
    (1,1,1,1), (2,1,1), (1,2,1), (1,1,2), (2,2), (1,3), (3,1)
  \end{align*}
 Because of the constraint $b(C) = 2$ all the chord diagrams we need to consider are
 \begin{align*}
   \rter{2,3,4}, \rter{2,4}, \rter{2,3}, \rter{2}
 \end{align*}
 For $\rter{2,3,4}$ and $\rter{2,4}$ only the decoration $d_1 = d_2 = d_3 = d_4 = 1$ is possible so we calculate these:
 We have one chord diagram with $ter = \{2,3,4\}$ and branch left vector $(0,1,1,1)$:
  \begin{align*}
    \ahat{C} = a_1^2 a_0 \text{ and } \omega_{\hat C} = (s-1)^2
  \end{align*}
There are three chord diagrams with $ter = \{2,4\}$, namely the two chord diagrams with branch left vector $(0,1,0,2)$:
   \begin{align*}
    \ahat{C} = a_2 a^2_0 \text{ and } \omega_{\hat C} = {s \choose 2}
  \end{align*}
 and the one with branch left vector $(0,1,1,1)$:
   \begin{align*}
    \ahat{C} = a_2 a^2_0 \text{ and } \omega_{\hat C} = (s-1)^2
  \end{align*}
 Summing this up $\rter{2,3,4}$ and $\rter{2,4}$ contribute to the left hand side by:
 \begin{align*}
      (s-1)^2a_1^2 a_0 + a_2 a^2_0 \left( 2{s \choose 2} + (s-1)^2 \right) = (s-1)^2 a_1^2a_0 + (2s^2-3s+1)a_2a_0^2
 \end{align*}
 $\rter{2,3}$ consists only of one chord diagram with branch left vector $(0,1,1)$ and we have to consider the decorations: $(d_1,d_2,d_3) \in \{ (2,1,1), (1,2,1), (1,1,2) \}$, so it contributes on the left hand side with:
 \begin{align*}
   \left( (2s-1) + 2(s-1) \right) a_0a_1 = (4s-3) a_0a_1
 \end{align*}
 $\rter{2}$ has only one chord diagram which contributes to left hand side by $3a_0$.
 The weight is $1$ because the branch left vector of the chord diagram is $(0,1)$ where the second coordinate is ignored by $\omega_{\hat C}$, but there are three chord diagrams to consider, namely those that are decorated by $(1,3)$,$(3,1)$ and $(2,2)$.
 So the left hand side of previous lemma is for this example:
 \begin{align*}
  (s-1)^2 a_1^2a_0 + (2s^2-3s+1)a_2a_0^2 + (4s-3) a_0a_1 + 3a_0
 \end{align*}
 For the right hand side we need to take the induced labels of $D_1$ and $D_2$ and we have to consider only the last sum:
 \begin{align*}
   \text{RHS} & = \left[ \sum_{\substack{||D_1|| = 1\\b(D_1)\geq 1}} \omega_{D_1} \ahat{D_1} a_{b(D_1)-1} \right]
\cdot \left[ \sum_{\substack{||D_2|| = 3\\b(D_2) = 1}} \omega_{\hat D_2} \ahat{D_2} \right] \\
  & + \left[ \sum_{\substack{||D_1|| = 2\\b(D_1)\geq 1}} \omega_{D_1} \ahat{D_1} a_{b(D_1)-1} \right]
\cdot \left[ \sum_{\substack{||D_2|| = 2\\b(D_2) = 1}} \omega_{\hat D_2} \ahat{D_2} \right] \\
  & + \left[ \sum_{\substack{||D_1|| = 3\\b(D_1)\geq 1}} \omega_{D_1} \ahat{D_1} a_{b(D_1)-1} \right]
\cdot \left[ \sum_{\substack{||D_2|| = 1\\b(D_2) = 1}} \omega_{\hat D_2} \ahat{D_2} \right] \\
 \end{align*}
 Since $N\in \mathbb N_{\geq 3}$, we have
 \begin{align*}
   \sum_{\substack{||D_2|| = 3\\b(D_2) = 1}} \omega_{\hat D_2} \ahat{D_2} = \sum_{\substack{||D_2|| = 2\\b(D_2) = 1}} \omega_{\hat D_2} \ahat{D_2} = \sum_{\substack{||D_2|| = 1\\b(D_2) = 1}} \omega_{\hat D_2} \ahat{D_2} = 1
 \end{align*}

 \begin{align*}
   & \sum_{\substack{||D_1|| = 1\\b(D_1)\geq 1}} \omega_{D_1} \ahat{D_1} a_{b(D_1)-1} 
\cdot = a_0 \\
   & \sum_{\substack{||D_1|| = 2\\b(D_1)\geq 1}} \omega_{D_1} \ahat{D_1} a_{b(D_1)-1} 
 = (s-1)a_0a_1 + a_0 \\
   & \sum_{\substack{||D_1|| = 3\\b(D_1)\geq 1}} \omega_{D_1} \ahat{D_1} a_{b(D_1)-1} 
 = (s-1)^2 a_0a_1^2 + (2s^3-3s+1)a_0^2a_2 + \left( (s-1) + (2s-1) \right) a_0a_1 + a_0
 \end{align*}
Summing this up, we indeed get the left hand side.
\end{example}
\begin{prop}[Restricted decorated version of Proposition 4.3 of \cite{MYchord}]
  \label{prop:restrictprop43}
  \begin{align*}
    \sum_{\substack{||C||=i+1\\b(C)=j+1\\\nu_{b(C)}=n}} \omega_{\hat C} \ahat{C} = \sum_{k=1}^i\sum_{l=1}^j {j \choose l} \left(
      \sum_{\substack{||D_1||=k\\b(D_1)\geq l}} \omega_{D_1}\ahat{D_1} a_{b(D_1)-l}\right)
    \left( \sum_{\substack{||D_2||=i-k+1\\b(D_2)= j-l+1\\\nu_{b(D_2)} = n-1}}
      \omega_{\hat D_2}\ahat{D_2}\right) 
  \end{align*}
\end{prop}
\begin{proof}
  Note that $\nu_{b(D_1\diamond D_2)} (D_1\diamond D_2) =
  \nu_{b(D_2)}(D_2) + 1$. Indeed, $b(D_1\diamond D_2) = b(D_2)$ and
  attaching a tree on the left side of $T(D_2)$ increments the rightmost
  branch, which ends at the leaf $b(D_2)$. Therefore the length of the rightmost branch matches as given in the proposition. The rest follows by Proposition \ref{prop:decprop43}
\end{proof}

The key idea in generalizing the results of this section from the results of \cite{MYchord} was to break up the sums in the lemmas above according to the branch left value of the base chord.  This further suggests that the branch left vector and the rooted trees are not just technical tools, but are showing something important about the structure of chord diagrams.

\section{Bridge equation and main result}\label{sec main}

In this section we develop an equation that builds a bridge between the sum that contains only chord diagrams and the sum that is calculated by derivations which will finally let us connect $G^{\text{dif}}$ and $G^{\text{comb}}$ (see Section~\ref{sec tree} for the definitions of $G^{\text{dif}}$ and $G^{\text{comb}}$). Therefore, we will call it the bridge equation. It is proved by induction.  The following lemma gives the base case whereas the lemma after it states the general case.
\begin{lemma}
  \label{lem:bridge-equation-special}Let $i\geq 1, j\geq 1$, then
  \begin{align*}\sum_{\substack{||C|| = i+1\\d_{j+1} = 1\\\nu_{j+1} =
        1\\b(C) = j+1} } \ahat{C} \omega_{\hat C} = [x^i]
    \left( \sum_{l\geq 1} \frac{g^{\text{comb}}_l(x)}{l!} \partial_{\rho=0}^l \right)
    \rho^j 
  \end{align*}
\end{lemma}
\begin{proof}
  Because the right hand side of the equation is only non-zero when $l=j$, it suffices to prove 
  \begin{align}
  \label{eq:lem5}
\sum_{\substack{||C|| = i+1\\d_{j+1} = 1\\\nu_{j+1} =
        1\\b(C) = j+1} } \ahat{C} \omega_{\hat C} = 
\sum_{\substack{||D|| = i\\b(D) \geq j} } \omega_{D}\ahat{D} a_{d(b(D)),b(D)-j}
  \end{align}
  Consider first $D$ satisfying the conditions on the right hand side of \eqref{eq:lem5}.  Let $C = D \diamond D_2$ with $D_2$ being the chord diagram with only one chord. We have the triangle inequality for
  the label of the base chords:  $b(D) \geq b(C) - b(D_2)$. Note that there indeed exists only one $C$ with this diamond decomposition and $b(C) = j+1$; this is because there is only one compatible shuffle which ends with the integer $j+1$. 
  Now, in the sum of the left hand side of \eqref{eq:lem5} every chord diagram $C$ splits into
  $D\diamond D_2$ with $D_2$ being the chord diagram with only one
  chord since $\nu_{j+1}=1$. Thus $D$ now has size $||D||=||C|| - d_{j+1} = i$ and its
  base chord satisfies $b(D) \geq j+1 -1 = j$. Thus the chord diagrams on each side of \eqref{eq:lem5} correspond via $C = D\diamond D_2$ and it is left to see
  that the corresponding summands are actually equal.
  This was already done in Lemma ~\ref{lem:declemma411}.
\end{proof}
Now, we can use this result to prove the statement that we are looking for:
\begin{lemma}[Bridge equation]Let $n\geq 1$, then
  \label{lem:bridge-equation}
\begin{align*}
  \sum_{\substack{||C|| = i+1\\d_{j+1} = 1\\\nu_{j+1} = n\\b(C)=j+1} } \ahat{C} \omega_{\hat C} = [x^i] \left( \sum_{l\geq 1} \frac{g^{\text{comb}}_l(x)}{l!} \partial_{\rho=0}^l \right)^n \rho^j
\end{align*}
\end{lemma}
\begin{proof}
  For better readability define: $G_{\partial \rho}(x) := \sum_{l\geq 1} \frac{g^{\text{comb}}_l(x)}{l!} \partial_{\rho=0}^l$
  and $F_{i,j,n} := [x^i] G_{\partial \rho}^n(x) \rho^j$. Let $i,j$ be fixed; we prove the statement by induction over $n$.
  For $n=1$ the statement is true by the previous Lemma.
  Now analogously to the proof of Lemma 4.14 in \cite{MYchord}, we observe that:
  \begin{align*}
    F_{i,j,n} & =[x^i] G_{\partial\rho}^n(x) \rho^j \\
    & = \sum_{k=1}^i \left( [x^k]G_{\partial\rho}(x) \right) \left( [x^{i-k}]G_{\partial\rho}^{n-1}(x) \right) \rho^j \\
    & \stackrel{\text{Leibniz-Rule}}{=}  \sum_{k=1}^i \sum_{l=1}^j {j \choose l}\left( [x^k]G_{\partial\rho}(x)\rho^l \right) \left( [x^{i-k}]G_{\partial\rho}^{n-1}  \rho^{j-l} \right) \\
    &  \stackrel{\text{Definition}}{=} \sum_{k=1}^i \sum_{l=1}^j {j \choose l} [x^k]g^{\text{comb}}_l(x) \cdot F_{i-k,j-l,n-1} \\
    &  \stackrel{\text{induction}}{=}\sum_{k=1}^i\sum_{l=1}^j {j \choose l} \left(
      \sum_{\substack{||D_1||=k\\b(D_1)\geq l}} \omega_{D_1}\ahat{D_1} a_{b(D_1)-l}\right)
    \left( \sum_{\substack{||D_2||=i-k+1\\b(D_2)= j-l+1\\\nu_{b(D_2)} = n-1}}
      \omega_{\hat D_2}\ahat{D_2}\right) \\
    & \stackrel{\text{Prop} ~\ref{prop:restrictprop43}}{=}\sum_{\substack{||C|| = i+1\\d_{j+1} = 1\\\nu_{j+1} = n} } \ahat{C} \omega_{\hat C} 
  \end{align*}
\end{proof}

The next result tells us that $g_1^{\text{comb}}$ satisfies the same recurrence as given by the Dyson-Schwinger equation itself for $g_1^{\text{dif}}$.
\begin{lemma}
  \label{thm:g1com-eq}
  Let $G^{\text{comb}}_{\partial\rho} :=  \sum_{l\geq 1} \frac{g^{\text{comb}}_l}{l!}$ and $\tilde F_k(\rho) := \sum_{l\geq 0}a_{k,l}\rho^l$,then 
  \begin{align*}
      g_1^{\text{comb}} = \sum_{k=1}^N x^k\sum_{n\geq 0} {n+sk-2 \choose n } (G^{\text{comb}}_{\partial\rho})^n \tilde F_k(\rho) 
  \end{align*}
\end{lemma}
\begin{proof}
  Consider the coefficient of $x^i$ on the right hand side of the equation that we want to show.
\begin{align*}
   [x^i] \text{RHS} & =  \sum_{k=1}^N \sum_{n\geq 0} {n+sk-2 \choose n } \sum_{l\geq 0} a_{k,l} [x^{i-k}](G^{\text{comb}}_{\partial\rho})^n \rho^l
 \end{align*}
Now, we can use the bridge equation and Lemma \ref{prop:big-restricted-sum-eq}:
\begin{align*}
  \sum_{\substack{||C|| = i+k\\d_{l+1} = k\\\nu_{l+1} = n\\b(C)=l+1} } \ahat{C} \omega_{\hat C} = [x^i] (G^{\text{comb}}_{\partial\rho})^n\rho^l \;\text{ for all } k = 1 \ldots N.
\end{align*}
So we get:
\begin{align*}
 [x^i] \text{RHS} & = \sum_{k=1}^N \sum_{n\geq 0} {n+sk-2 \choose n } \sum_{l\geq 0} a_{k,l} \sum_{\substack{||C|| = i\\d_{l+1} = k\\\nu_{l+1} = n\\b(C)=l+1} } \ahat{C} \omega_{\hat C} \\
    & = \sum_{k=1}^N\sum_{n\geq 0}  \sum_{l\geq 0} a_{k,l} \sum_{\substack{||C|| = i\\d_{l+1} = k\\\nu_{l+1} = n\\b(C)=l+1} } \ahat{C} \omega_{\hat C} {\nu_{l+1}+sk-2 \choose \nu_{l+1} } \\
    & = \sum_{k=1}^N\sum_{n\geq 0}  \sum_{l\geq 0} \sum_{\substack{||C|| = i\\d_{b(C)} = k\\\nu_{b(C)} = n\\b(C)=l+1} }  a_{k,b(C)-1}\ahat{C} \omega_{\hat C} {\nu_{l+1}+sk-2 \choose \nu_{l+1} } 
\end{align*}
Since we have $l = b(C) + 1$, we see that 
\begin{align*}
  \omega_{\hat C} {\nu_{l+1}+sk-2 \choose \nu_{l+1} } = \omega_C
\end{align*}
Look carefully at the restrictions of the last sum. We need to verify that we can drop the last three constraints because we are summing over all possible $k,n,l$:
\begin{enumerate}
\item The restriction $d_{b(C)} = k$ drops because we are summing over all $k = 1\ldots N$. 
\item The restriction of $\mathcal R_{dec}$ to $||C|| = i, \nu_{b(C)} = n$ is always non-empty for some $n$ and summing over all $n$ indeed yields all rooted connected decorated chord diagrams with $||C|| = i$.
\item $b(C) \geq 1$ so we can drop $b(C) = l+1$ and the sum over $l$.
\end{enumerate}
In conclusion
\begin{align*}
[x^i]\text{RHS} = \sum_{||C|| = i} \omega_C\ahat{C} a_{d(b(C)),b(C)-1} = g_1^{\text{comb}}
\end{align*}
\end{proof}


\begin{thm}\label{main thm}
The analytic Dyson-Schwinger equation 
\begin{align*}
 & G(x,L) = 1 - \sum_{k \geq 1} x^kG(x,\partial_{-\rho = 0})^{1-sk} (e^{-L\rho}-1)F_k(\rho) \text{ where } F_k(\rho) = \sum_{l\geq 0} a_{k,l}\rho^{l-1}\\   
\end{align*}
has as formal solution the following combinatorial expansion in terms of chord diagrams:
\begin{align*}
  G(x,L) = 1 - \sum_{k\geq 1} \frac{(-L)^k}{k!} \sum_{b(C)\geq k} \omega_C \ahat{C} a_{d_{b(C)},b(C)-k}x^{||C||}
\end{align*}
\end{thm}
\begin{proof}
Since the renormalization group equation is true for both $g^{\text{comb}}$ (by Theorem \ref{thm:rge-equation}) and $g^{\text{dif}}$ (by Proposition \ref{prop:diff rg eqn}) , $G^{\text{dif/comb}}$ are built from $g^{\text{dif/comb}}_1$ in completely the same way:
  \begin{align*}
    g^{\text{dif/comb}}_k = g_1^{\text{dif/comb}} (sx\partial_x - 1)g_{k-1}^{\text{dif/comb}}  \; \text{ for } k \geq 2
  \end{align*}
For $g_1^{\text{dif}}$, by applying the generalized geometric series to $G^{\text{dif}}(x,\partial_{-\rho})^{1-sk}$ in the Dyson-Schwinger equation we have:
\begin{align*}
  g_1^{\text{dif}} = \sum_{k=1}^\infty x^k\sum_{n\geq 0} {n+sk-2 \choose n } (G^{\text{dif}}_{\partial\rho})^n \tilde F_k(\rho) 
\end{align*}
where $\tilde F_k(\rho) = (e^{-L\rho}-1)\frac{1}{\rho}\sum_{l\geq 0}a_{k,l}\rho^l$. 
 So, by Lemma \ref{thm:g1com-eq} for $g_1^{\text{comb}}$ and the Dyson-Schwinger equation itself for $g_1^{\text{dif}}$
\begin{align*}
  g_1^{\text{comb}}(x) - g_1^{\text{dif}}(x) = \sum_{k=1}^N x^k\sum_{n\geq 0} {n+sk-2 \choose n } \left( (G^{\text{comb}}_{\partial\rho})^n - (G^{\text{dif}}_{\partial\rho})^n \right)\tilde F_k(\rho) 
\end{align*}
Write $g_1^{\text{comb}} = \sum_{k\geq 1} g_{1,k}x^k$ and similarly for $g_1^{\text{dif}}$.  
The first coefficients of $g_1^{\text{comb}}$ and $g_1^{\text{dif}}$ are the same.
Apply the renormalization group equation to convert the bigger $g_k$ to $g_1$; note that the coefficient of $x^k$ in the equation above only involves $g_{1,\ell}$ with $\ell < k$ on the right hand side.  So inductively we obtain $g^{\text{comb}}_{1,k} = g^{\text{dif}}_{1,k}$ and so 
\begin{align*}
  g_1^{\text{comb}}(x) - g_1^{\text{dif}}(x) = 0
\end{align*}
\end{proof}

\section{Conclusion}

What we have achieved in our main theorem, Theorem~\ref{main thm} is to solve a family of Dyson-Schwinger equations as expansions over decorated rooted connected chord diagrams.  The analytic input to the expansion is the expansion of the regularized integrals for the primitive graphs building the Dyson-Schwinger equation.  Each decorated chord diagram in the expansion contributes a weighted monomial in the coefficients of the expansions of the primitives.

Compared to \cite{MYchord} the achievement here has been the generalization to a substantially larger class of Dyson-Schwinger equations.  This indicates that these chord diagram expansions are not mere curiosities but actually quite general.  Whether or not this would be the case was one of the main open questions after \cite{MYchord}.  This generalization also clarifies some of the technical underpinnings, in particular showing the key role that the binary rooted tree associated to a chord diagram plays.

Furthermore, in the course of these investigations one of us (MH) wrote code to calculate chord diagrams, their weights and all the associated objects which we define.  This code is available in \cite{Hphd}.

\medskip

There are a few obvious things to consider next.  First would be a generalization to systems of Dyson-Schwinger equations which would capture more of the Dyson-Schwinger equations of interest in physics.  Second is to consider the asymptotic consequences of these results.  There is a reasonably good understanding of the asymptotics of chord diagrams, but not of the particular parameters important here.  This is being investigated in the simpler case of \cite{MYchord} by the second author with Julien Courtiel.  Third there are unresolved combinatorial issues.  For example the rooted tree construction is not very natural combinatorially, but its importance suggests that it should be.  Therefore, there ought to be a more transparent reformulation of the rooted trees.  There are also many patterns in the coefficients of our chord diagram expansion which have barely been investigated.  See \cite{Hphd} for more details of some of these patterns.

\bibliographystyle{plain}
\bibliography{main}

\begin{thebibliography}{10}

\bibitem{BKerfc}
D.J. Broadhurst and D.~Kreimer.
\newblock Exact solutions of {D}yson-{S}chwinger equations for iterated
  one-loop integrals and propagator-coupling duality.
\newblock {\em Nucl. Phys. B}, 600:403--422, 2001.
\newblock arXiv:hep-th/0012146.

\bibitem{ck0}
Alain Connes and Dirk Kreimer.
\newblock Hopf algebras, renormalization and noncommutative geometry.
\newblock {\em Commun. Math. Phys.}, 199:203--242, 1998.
\newblock arXiv:hep-th/9808042.

\bibitem{e-fksurvey}
Kurusch Ebrahimi-Fard and Dirk Kreimer.
\newblock Hopf algebra approach to {F}eynman diagram calculations.
\newblock {\em J. Phys. A}, 38:R285--R406, 2005.
\newblock arXiv:hep-th/0510202.

\bibitem{Fdse}
Lo\"ic Foissy.
\newblock Fa\`a di {B}runo subalgebras of the {H}opf algebra of planar trees
  from combinatorial {D}yson-{S}chwinger equations.
\newblock {\em Advances in Mathematics}, 218(1):136--162, 2007.
\newblock arXiv:0707.1204.

\bibitem{Fphysical}
Lo\"ic Foissy.
\newblock General {D}yson-{S}chwinger equations and systems.
\newblock {\em Communications in Mathematical Physics}, 327(1):151--179, 2014.
\newblock arXiv:1112.2606.

\bibitem{Hphd}
Markus Hihn.
\newblock {\em The generalized chord diagram expansion}.
\newblock PhD thesis, Humboldt University of Berlin, (submitted 2015).

\bibitem{MYchord}
Nicolas Marie and Karen Yeats.
\newblock A chord diagram expansion coming from some {D}yson-{S}chwinger
  equations.
\newblock {\em Communications in Number Theory and Physics}, 7(2):251--291,
  2013.
\newblock arXiv:1210.5457.

\bibitem{vS}
Walter~D. van Suijlekom.
\newblock Renormalization of gauge fields: A {H}opf algebra approach.
\newblock {\em Commun. Math. Phys.}, 276:773--798, 2007.
\newblock arXiv:hep-th/0610137.

\bibitem{vS2}
Walter~D. van Suijlekom.
\newblock Renormalization of gauge fields using {H}opf algebras.
\newblock In B.~Fauser, J.~Tolksdorf, and E.~Zeidler, editors, {\em Recent
  Developments in Quantum Field Theory}. Birkhauser Verlag, 2008.
\newblock arXiv:0801.3170v1.

\bibitem{Ymem}
Karen Yeats.
\newblock Rearranging {D}yson-{S}chwinger equations.
\newblock {\em Mem. Amer. Math. Soc.}, 211, 2011.

\bibitem{kythesis}
Karen~Amanda Yeats.
\newblock {\em Growth estimates for {D}yson-{S}chwinger equations}.
\newblock PhD thesis, Boston University, 2008.

\end{thebibliography}

\end{document}